\documentclass[11pt]{article}

\usepackage[utf8]{inputenc}
\usepackage[T1]{fontenc}
\usepackage{amsmath,amssymb,amsthm,mathtools}
\usepackage[margin=1in]{geometry}
\usepackage[numbers,sort&compress]{natbib}
\usepackage{microtype}
\usepackage{enumitem}
\usepackage{booktabs}
\usepackage{graphicx}
\usepackage{hyperref}
\hypersetup{colorlinks=true,linkcolor=black,citecolor=black,urlcolor=black}

\newcommand{\dd}{\mathrm{d}}
\newcommand{\ii}{\mathrm{i}}
\newcommand{\Ptot}{P_{\mathrm{tot}}}
\newcommand{\Pcap}{P_{\mathrm{cap}}}
\renewcommand{\Phi}{\varphi} % \Phi existe déjà (grand Phi) -> renew vers varphi
\newcommand{\Kc}{K_{\mathrm c}}

\theoremstyle{definition}
\newtheorem{assumption}{Assumption}
\theoremstyle{plain}
\newtheorem{proposition}{Proposition}
\newtheorem{corollary}{Corollary}

\title{Do Co-Located AI Training Jobs Synchronize?\\
Load-Dependent Throttling as a Coupling Mechanism for\\
Phase-Locking Behind a Shared Power Cap}

\author{Brieuc Le Roux Tardif\\ \small IMT Nord Europe}
\date{\today}

\begin{document}
\maketitle

%==============================================================================
\begin{abstract}
Large-scale artificial-intelligence training has turned individual computing
facilities into multi-megawatt loads whose power draw is not constant but
periodic: tens of thousands of accelerators step in lockstep between
compute-bound phases, in which they consume near peak power, and
communication-bound phases, in which they idle while gradients are exchanged.
A growing body of work treats each such facility as an exogenous periodic
forcing on the electrical grid and mitigates on the grid side with storage, or
on the facility side by smoothing each job's power profile individually. This
paper poses instead the question that belongs to the facility
operator: when many independently operated training jobs share one
oversubscribed power envelope, do their power cycles remain statistically
independent -- so that their swings partially cancel and the aggregate
fluctuation grows only as the square root of the number of jobs -- or can the
operator's own power-management stack phase-lock them, so that the fluctuation
grows linearly and the site swings coherently against its own caps, its
power-delivery hardware, and its interconnection agreement? Posed this way, the
problem is one of emergent synchronization in a population of nonlinear
oscillators, the setting of the Kuramoto paradigm. The analogy, however, hides a
missing ingredient: classical synchronization requires a physical channel through
which each oscillator's phase velocity depends on the others' phases, and grid
frequency does not provide it, because accelerator clocks are decoupled
from the line by rectification and regulation. We identify a candidate coupling
channel in load-dependent throttling: facility- and rack-level power
oversubscription with enforced caps, voltage droop, and shared cooling all slow
computation precisely when aggregate demand is high, making each job's iteration
time a function of every other job's instantaneous power -- a channel that is
owned, configured, and observable by the operator. Formalizing the fleet as a
generalized Kuramoto system whose coupling function is fixed by the
(non-sinusoidal, relaxation-type) training waveform convolved with the
throttling response, we obtain three operator-facing statements.
Dimension: the coupling is, to leading order, repulsive -- peaks
repel, and fast electrical power capping actively spreads the fleet -- and it
turns attractive only when the control loop's total phase lag at the iteration
frequency exceeds half a cycle. The lag budget discriminates between channels:
analog lag stages (sensing filters, a single thermal stage) each contribute a
phase bounded by $\pi/2$ and cannot invert the sign on their own, whereas dead
time (polling and enforcement intervals, sample-and-hold, actuation latency,
cascaded slow stages) accumulates phase without bound and can; hard saturation
of the cap is a separate, non-perturbative hazard. An order-parameter criterion
then gives the critical coupling in terms of the spread
of natural iteration rates and the loop lag. The protection is
mode-selective, however: the sign law is per-harmonic, and a sufficiently uniform
fleet locks into a higher-harmonic cluster state -- an aggregate that again grows
linearly in the number of jobs, at a multiple of the iteration frequency -- so
rate diversity is a requirement that no choice of control delay can replace.
Detect: the reduction's
own frequency-correlated frustration renders the onset first-order and
hysteretic for a sufficiently diverse fleet -- an intrinsic mechanism requiring
no inertia, exhibited in direct finite-$N$ simulation and consistent with an
exact mean-field (Ott--Antonsen) onset computation -- so a nominally safe fleet can be
tipped into, and remain stuck in, the locked state; the order parameter and its
first few harmonics, computable from per-job power telemetry, are the quantities
to watch after transients. Mitigate: phase-scattering scheduling -- the deliberate
detuning and offsetting of training jobs -- raises every mode threshold at once,
at a throughput cost we estimate at first order. The study is analytical and
simulation-based, with no facility data; its sign prediction is directly
falsifiable by a two-job co-capped measurement, which we specify. The
consequences of a locked fleet for the grid, a possible forcing of inter-area
modes, are noted as an extension and deferred to companion work.
\end{abstract}

%==============================================================================
\section{Introduction}\label{sec:intro}

\paragraph{The phenomenon.}
The power demand of a large training job is structured in time, not constant. Each
optimization step alternates a compute-bound interval, in which accelerators run
near their thermal design power, with a communication-bound interval, in which they
stall on collective operations and fall to a much lower draw
\citep{megascale2024,lee2025overlap}. Because the step repeats near-identically
for the duration of a run, the facility presents the grid with a quasi-periodic
load whose fundamental period is the iteration time $T$ (seconds to tens of
seconds) and whose peak-to-peak amplitude, at hyperscale, reaches tens to hundreds
of megawatts with ramp rates exceeding $10~\mathrm{MW\,s^{-1}}$
\citep{ko2025widearea,oracle2025smoothing}. Two trends make this a present concern
rather than a curiosity: single-campus power is growing toward the gigawatt scale \citep{masanet2020},
and a campus is increasingly filled with many near-identical jobs --
replicated foundation-model runs and large fleets of fine-tunes -- so that one
iteration waveform is reproduced across the site.

\paragraph{From the grid's problem to the operator's question.}
The existing literature studies a facility as a single exogenous oscillation that
forces the grid and can excite inter-area and sub-synchronous modes
\citep{ko2025widearea,forced2025,microsoft2025powerstab}; the proposed mitigations
act on the grid side through energy storage \citep{hybridess2025} or on the
facility side by reshaping each job's power profile individually -- workload
insertion and millisecond-scale GPU power smoothing
\citep{microsoft2025powerstab,oracle2025smoothing}. We
take the other seat. Before the aggregate swing ever reaches a grid mode it must
traverse the operator's own infrastructure: the oversubscribed power envelope and
the caps that enforce it, power-delivery hardware sized for a smoother profile,
and the interconnection agreement that bounds ramp and swing at the point of
common coupling. Seen from that seat, the grid literature's implicit assumption --
that when $N$ jobs coexist their forcing terms superpose with statistically
independent phases -- is not a modelling convenience but a capacity-planning
premise: under it the aggregate fluctuation grows only as $\sqrt N$ -- partial
cancellation -- stays benign relative to the mean load even at large $N$, and
oversubscription is safe. The question of this paper is whether the premise
holds, and it is the operator's question:
can the jobs co-located behind one power envelope spontaneously phase-lock
through the site's own power-management stack -- and if so, under what conditions,
with what warning, and at what cost to prevent?
Oscillators that share a medium can synchronize; if the jobs phase-lock, the
coherent part of the aggregate grows as $N$, not $\sqrt N$, and the site swings
against its own envelope with a relative amplitude larger by a factor
$\sim\!\sqrt N$ (Section~\ref{sub:scaling}). Whether the realistic regime is
the benign one or the dangerous one is therefore not a detail but the whole
question, and it is a question about emergent synchronization in a population of
nonlinear oscillators -- the Kuramoto setting
\citep{winfree1967,kuramoto1975,strogatz2000,acebron2005,pikovsky2001}.

\paragraph{The coupling is not the obvious one.}
Synchronization requires a physical channel through which each oscillator's phase
velocity depends on the others' phases. The obvious candidate -- grid frequency, by
analogy with networks of synchronous machines
\citep{filatrella2008,dorfler2013,motter2013} -- does not serve: accelerator step
rates are set by on-board clocks downstream of rectification and regulation and are
insensitive, to first order, to line-frequency variation (a voltage sag does slow
computation, but through the power-delivery and throttling channel of
Assumption~\ref{ass:clock}, not as a frequency entrainment), so a model
coupling jobs through grid frequency is physically vacuous. We identify the
operative channel as load-dependent throttling: shared, congestible
resources -- power oversubscription with enforced caps, voltage droop, and shared
cooling -- slow computation precisely when aggregate demand is high, making each
job's step rate a decreasing function of every other job's instantaneous power
(Section~\ref{sub:coupling}; Figure~\ref{fig:mech}). This closes the loads into a self-consistent
dynamical system and supplies the missing coupling. Identifying it correctly is
what turns a loose analogy into a falsifiable mechanism, and it reverses the naive
intuition, because throttling-on-aggregate-load is, by default, repulsive
rather than attractive. It also relocates the problem: every channel involved --
cap policy, control-loop bandwidth, cooling -- is owned and configurable by the
operator, so the same analysis that identifies the hazard identifies the levers.

\begin{figure}[t]\centering
\includegraphics[width=.8\linewidth]{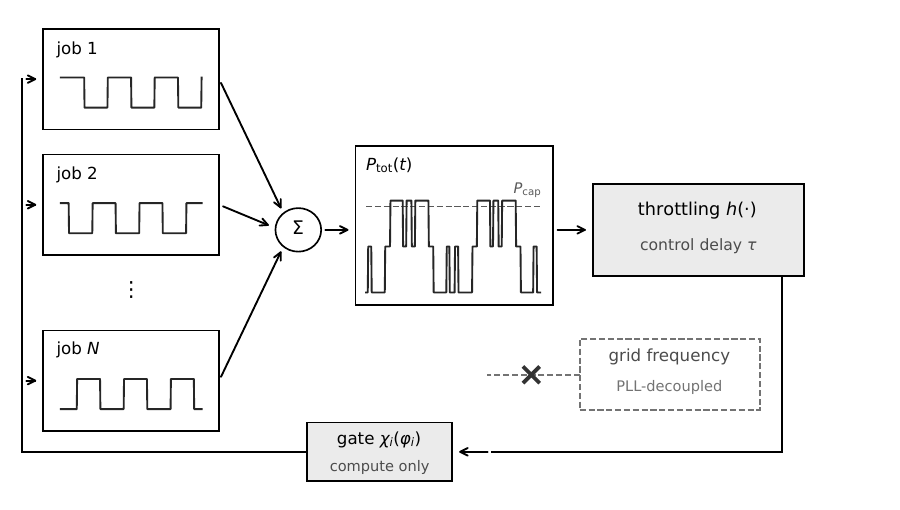}
\caption{The coupling mechanism, schematically (formalized in
\eqref{eq:waveform}--\eqref{eq:throttle}). Each job draws the two-level power
waveform of \S\ref{sub:waveform}; the aggregate $\Ptot$ is compared with the
shared cap $\Pcap$, and the excess drives the throttling response $h$ --
lumping enforced power capping, voltage droop, and shared cooling -- after a
control delay $\tau$. The throttle slows a job only through the phase gate
$\chi_i(\Phi_i)$, active while the job is compute-bound: this gating is what
turns a common drive into a genuine phase coupling (\S\ref{sub:reduction}).
Grid frequency (crossed channel) cannot reach the accelerator clocks
and couples nothing (Assumption~\ref{ass:clock}).}
\label{fig:mech}
\end{figure}

\paragraph{Contributions.}
The paper answers the operator's question in the order an operator would ask it --
what is at stake, through what mechanism, at what threshold, with what warning,
and with what remedy -- and we are explicit about the epistemic status
of each contribution: exact consequence of the definitions, leading-order
analytical result, or claim to be established numerically.
\begin{enumerate}[leftmargin=1.6em,itemsep=3pt]
  \item \textbf{The stake, made quantitative.} We argue the consequential
  question is collective, not individual, and we make the stake precise:
  synchronization changes the aggregate scaling from $\sqrt N$ to $N$, an
  amplification of the relative swing by the coherence amplification factor
  $r\sqrt N$ (Section~\ref{sub:scaling}), where $r$ is the order parameter. This is
  an exact consequence of the definitions.
  \item \textbf{Identification of the coupling mechanism -- an operator-owned one.}
  We show that grid frequency cannot couple the jobs and that load-dependent
  throttling can, and we write the resulting closed load dynamics
  \eqref{eq:throttle}. This is a modelling result; its premises
  (Assumptions~\ref{ass:domain}--\ref{ass:throttle}) are stated, justified, and
  empirically checkable -- and every element of the mechanism (cap policy,
  control delay, cooling) is a configuration choice of the site, not of the grid.
  \item \textbf{Reduction to a generalized Kuramoto model.} Linearizing the
  throttling response yields a Sakaguchi--Daido coupling whose function is fixed by
  the training waveform convolved with the throttling response
  \eqref{eq:reduced}--\eqref{eq:Gamma-sign}. At leading order this gives three
  statements (Corollaries~\ref{cl:repulsion}--\ref{cl:attraction} and
  Proposition~\ref{cl:clusters}): default repulsion of
  the fundamental at short control delay; a sign reversal to attraction beyond a
  critical delay $\tau^\star=T/2$; and a per-harmonic sign law
  $a^{(m)}\propto-\sin(m\bar\omega\tau)$ by which a near-identical fleet locks into
  higher-harmonic cluster states even at fundamental-repulsive delays, with the
  mode thresholds \eqref{eq:kcm}. These are analytical results valid in the
  weak-coupling, slow-phase regime (Assumption~\ref{ass:slowphase}), each tested
  against the unaveraged model (Section~\ref{sec:numerics}).
  \item \textbf{A dimensioning criterion, and the warning that qualifies it.} An
  order-parameter self-consistency argument gives the onset of fleet-wide coherence
  and the critical coupling
  $\Kc=2/[\pi\,a_1(\tau)\,g(0)]$ (Proposition~\ref{prop:kc}), from which
  heterogeneity of iteration rates raises $\Kc$ while near-identical fleets collapse
  it (Corollary~\ref{cor:hetero}) -- a criterion an operator could evaluate from the cap
  stiffness, the loop lag, and the job mix, given measured waveforms
  (\S\ref{sub:waveform}); an exact Ott--Antonsen computation for a Lorentzian
  fleet gives the averaged model's onset in closed form, of which
  \eqref{eq:kc-general} is the leading-order limit -- and a fine sweep of the
  full model lands on the leading-order value
  (Appendix~\ref{app:oa}, \S\ref{sub:num-robust}). The reduction also carries a
  frequency-correlated frustration $\alpha_j=\omega_j\tau+\pi/2$, which
  drives a first-order (subcritical, hysteretic) onset for diverse fleets
  through a twisted order parameter -- an intrinsic mechanism needing no inertia
  (Section~\ref{sub:order}), exhibited by direct simulation
  (Section~\ref{sec:numerics}) and structurally consistent with the Ott--Antonsen
  normal form (Appendix~\ref{app:oa}). Hysteresis converts the criterion from a static
  sizing rule into a monitoring obligation: a fleet below threshold can be tipped
  into the locked state and stay there, so the order parameter -- computable from
  per-job power telemetry -- must be watched after transients
  (Section~\ref{sec:discussion}). The threshold is analytical; the order of the
  transition reduces, on the Ott--Antonsen manifold, to the sign of a Landau
  coefficient whose explicit evaluation we leave open (Appendix~\ref{app:oa}) --
  our evidence for the first-order onset is the finite-$N$ continuation, which
  brackets the supercritical/subcritical crossover between two tested spreads.
  \item \textbf{A software-only mitigation.} We cast phase-scattering
  scheduling as the deliberate engineering of repulsive coupling -- detuning
  iteration rates and imposing fixed phase offsets -- with a first-order throughput
  cost \eqref{eq:penalty}. It needs no storage and lives at the scheduler the
  operator already runs. Existing software mitigations shape each job's power
  profile individually \citep{microsoft2025powerstab,oracle2025smoothing};
  phase-scattering acts instead on the relative phases of the fleet -- a
  complementary, collective lever.
\end{enumerate}

\paragraph{Scope and falsifiability.}
The analysis is deliberately a leading-order, mean-field theory: a single shared
power domain (all-to-all coupling), a two-level waveform, and retention of the
fundamental harmonic. These choices are recorded as
Assumptions~\ref{ass:waveform}--\ref{ass:slowphase}, with the cost of each noted
where it is incurred. The theory is falsifiable in two concrete ways. It predicts a
sign for the coupling -- repulsive at small control delay -- that a
controlled measurement of two co-capped jobs would confirm or refute. And it
predicts that the dangerous regime is reached by increasing control delay or by
homogenizing iteration rates, not by adding jobs at fixed heterogeneity
behind a fixed total cap -- the qualification matters, because the bare
coupling scale is extensive: it is the saturation of a fixed cap that keeps $K$
intensive (Assumption~\ref{ass:slowphase}), and at fixed per-job cap stiffness
adding jobs raises $K$ linearly and can cross $\Kc$ on its own.
Should fleet-wide synchronization prove marginal under realistic parameters, the
same apparatus yields the complementary and equally useful statement -- the
conditions under which, within the model, co-located jobs do not
synchronize -- a design criterion for operators rather than a null result.

\paragraph{Outline.}
Section~\ref{sec:prelim} recalls the synchronization tools. Section~\ref{sec:model}
states the standing assumptions and notation (\S\ref{sub:assumptions}) and builds
the model: the per-job waveform (\S\ref{sub:waveform}), the coupling mechanism
(\S\ref{sub:coupling}), and the reduction to a generalized Kuramoto system
(\S\ref{sub:reduction}). Section~\ref{sec:analysis} develops the order-parameter
analysis, the critical coupling, the role of inertia and delay, and
phase-scattering scheduling. Sections~\ref{sec:related}--\ref{sec:conclusion} place
the work in context, give the numerical protocol, and discuss limitations.

%==============================================================================
\section{Preliminaries}\label{sec:prelim}
% >>> RÉDIGÉ. Standard et citable ; indépendant de ta direction.

\subsection{The Kuramoto paradigm and the order parameter}\label{sub:kuramoto}
Consider $N$ phase oscillators $\theta_i\in[0,2\pi)$, each with an intrinsic
frequency $\omega_i$ drawn from a density $g(\omega)$, coupled all-to-all:
\begin{equation}\label{eq:kuramoto}
  \dot\theta_i \;=\; \omega_i \;+\; \frac{K}{N}\sum_{j=1}^{N}\sin(\theta_j-\theta_i),
  \qquad i=1,\dots,N .
\end{equation}
The collective state is summarized by the complex order parameter
\begin{equation}\label{eq:order}
  r(t)\,e^{\ii\psi(t)} \;=\; \frac{1}{N}\sum_{j=1}^{N} e^{\ii\theta_j(t)},
  \qquad r\in[0,1],
\end{equation}
so that \eqref{eq:kuramoto} closes into the mean-field form
$\dot\theta_i=\omega_i+Kr\sin(\psi-\theta_i)$: each oscillator feels the mean
field only through its amplitude $r$ and phase $\psi$. The modulus $r$ measures
coherence: $r\approx 0$ in the incoherent state (phases uniformly spread) and
$r\to 1$ as the population locks. We will also need the higher Daido moments
$r_m e^{\ii\psi_m}=N^{-1}\sum_j e^{\ii m\theta_j}$ ($r\equiv r_1$), which measure
coherence mode by mode and detect the cluster states of
\S\ref{sub:generalized}. For a unimodal, symmetric $g$, Kuramoto's
self-consistency argument gives a continuous transition at
\begin{equation}\label{eq:kc}
  \Kc \;=\; \frac{2}{\pi\,g(0)} ,
\end{equation}
below which $r=0$ in the thermodynamic limit and above which a macroscopic
synchronized cluster emerges \citep{kuramoto1975,kuramoto1984,strogatz2000,acebron2005}.
Two facts will matter repeatedly. First, the relevant scale of $\Kc$ is set by
the width of $g$: a more heterogeneous population (broader spread of
natural rates) is harder to synchronize. Second, below $\Kc$ the order parameter
does not vanish at finite $N$ but fluctuates with $r\sim N^{-1/2}$; this residual
coherence is the seed of the $\sqrt{N}$ scaling we revisit in
\S\ref{sub:scaling}.

\subsection{Second-order oscillators and grid synchronization}\label{sub:inertia}
Synchronous machines on a transmission network obey the swing equation, which is
a phase oscillator with inertia and damping,
\begin{equation}\label{eq:swing}
  M_i\,\ddot\theta_i + D_i\,\dot\theta_i
  \;=\; P_i^{\mathrm{mech}} \;-\; \sum_{j} a_{ij}\sin(\theta_i-\theta_j),
\end{equation}
the structure that underlies the Kuramoto-like models of power grids
\citep{bergenhill1981,filatrella2008,dorfler2013,motter2013,dorflerbullo2014,nishikawa2015,rohden2012}. Two consequences of the inertial
term $M_i\ddot\theta_i$ are important here. (i) The synchronization transition can
become discontinuous (first-order, ``explosive''), with a hysteresis loop
between the incoherent and locked branches, when inertia or a frequency--degree
correlation is present \citep{tanaka1997,rodrigues2016}. (ii) Inertia introduces a natural
timescale that can resonate with the forcing period of a training waveform; the
relevant comparison is between the iteration rate $\omega_i$ and the grid's
electromechanical modes. We will reuse \eqref{eq:swing} not for the data-center
loads themselves -- which carry negligible mechanical inertia -- but for the grid
nodes to which they attach, and which mediate the coupling of
\S\ref{sub:coupling}.

\subsection{Beyond sinusoidal coupling: generalized and pulse-like coupling}
\label{sub:generalized}
The sinusoid in \eqref{eq:kuramoto} is a modelling convenience, exact only in the
weak-coupling limit of smooth limit-cycle oscillators near a Hopf bifurcation.
Training loads are not of this type: their power trace is closer to a
relaxation oscillation -- a near-square alternation between two power
levels -- so the interaction inherits the full harmonic content of the waveform.
The appropriate generalization replaces the sine by an arbitrary
$2\pi$-periodic odd-plus-even coupling function $\Gamma$,
\begin{equation}\label{eq:generalized}
  \dot\theta_i \;=\; \omega_i \;+\; \frac{K}{N}\sum_{j=1}^N \Gamma(\theta_j-\theta_i),
  \qquad
  \Gamma(x)=\sum_{m\ge 1}\bigl[a_m\sin(mx)+b_m\cos(mx)\bigr],
\end{equation}
the Daido / Sakaguchi--Kuramoto family \citep{sakaguchi1986,daido1996,nakao2016}. Three
features of \eqref{eq:generalized} drive the rest of the paper.
\begin{itemize}[leftmargin=1.4em,itemsep=2pt]
  \item \textbf{Sign of the first harmonic.} If $a_1>0$ the coupling is
  attractive and favours in-phase locking; if $a_1<0$ it is
  repulsive and favours anti-phase, splay, or cluster states in which the
  population spreads itself out and the order parameter stays low. The sign of
  $a_1$ is therefore the most consequential quantity for the fundamental
  mode, deciding between coherent ($r\to1$) and incoherent ($r\to0$) aggregates
  at the iteration frequency -- with the caveat, developed in
  Proposition~\ref{cl:clusters}, that each higher harmonic carries its own sign,
  so a fundamental-repulsive fleet can still lock at a higher mode.
  \item \textbf{Phase lag.} A constant lag, $\Gamma(x)=\sin(x-\alpha)$, mixes the
  $\sin$ and $\cos$ harmonics; as $\alpha$ passes $\pi/2$ the effective $a_1$
  changes sign. Lags arise physically from delay in whatever channel
  mediates the coupling -- which, in \S\ref{sub:coupling}, is a control loop
  \citep{yeung1999}.
  \item \textbf{Higher harmonics.} A hard, saturating interaction (here, a power
  cap that clips rather than bends) injects $m\ge2$ harmonics, which stabilize
  multi-cluster states -- several synchronized sub-groups rather than one
  \citep{hansel1993}. These
  matter for the shape, not merely the amplitude, of the aggregate swing.
\end{itemize}
For completeness we note the pulse-coupled limit
\citep{mirollo1990}, in which interaction is concentrated at phase boundaries; it
is the singular case of \eqref{eq:generalized} as the communication phase narrows,
and we return to it in \S\ref{sub:reduction}.

%==============================================================================
\section{From training clusters to coupled oscillators}\label{sec:model}
% >>> RÉDIGÉ. C'est le coeur de modélisation et la partie la plus indépendante
%     de ta direction : une dérivation.

\subsection{Standing assumptions and notation}\label{sub:assumptions}
We model a site running $N$ training jobs indexed $i=1,\dots,N$. Job $i$ carries an
internal phase $\Phi_i\in[0,2\pi)$ advancing through one optimization step, a
natural angular rate $\omega_i=2\pi/T_i$ with $T_i$ its iteration time, and an
instantaneous power $P_i(\Phi_i)$; the coherence of the population is measured by
the order parameter $r$ of \eqref{eq:order}. The following assumptions hold
throughout. Each is stated with its physical justification and with the cost of
relaxing it, so that the regime of validity of every later result is explicit.

\begin{assumption}[Two-level waveform with bounded overlap]\label{ass:waveform}
Over the analysed horizon each job runs a fixed model and parallelization plan, so
$P_i$ is periodic in $\Phi_i$; we approximate it by the two levels
\eqref{eq:waveform}. Justification: steady-state training repeats one step for
$10^{4}$--$10^{6}$ iterations, with measured periods $T\sim2$--$6$\,s and a clear
compute/communication power contrast \citep{megascale2024,ko2025widearea}.
Cost and caveat: modern frameworks deliberately overlap collectives
with computation \citep{lee2025overlap}, which raises $P_i^{\mathrm{lo}}$, lowers
the contrast $A=P_i^{\mathrm{hi}}-P_i^{\mathrm{lo}}$, and so shrinks the fundamental
$c_{i,1}\propto A$ and the coupling $K\propto c_1$ (\S\ref{sub:reduction}). The
upper-bound status of the square waveform must be stated per harmonic: for
any waveform confined to $[P^{\mathrm{lo}},P^{\mathrm{hi}}]$ the fundamental
obeys $c_1\le 2A/\pi$, the value attained by the $\delta=\tfrac12$ square (the
extremal waveform is bang-bang), so the square is a genuine upper bound at $m=1$;
at $m\ge2$ the square envelope $2A/(\pi m)$ bounds only waveforms with a single
compute pulse per iteration, and multi-pulse structure (micro-batching, several
collectives per step) can exceed it -- the harmonic-lock thresholds
\eqref{eq:kcm} are correspondingly waveform-specific. Finite ramps are handled in
Appendix~\ref{app:fourier}. Susceptibility caveat: the gate
\eqref{eq:susceptibility} idealizes the communication phase as cap-insensitive;
in reality collectives run kernels on the accelerators and some communication
paths are clock-sensitive \citep{lee2025overlap}, so the physical susceptibility
has $\chi^{\mathrm{lo}}>0$ and its own waveform. A reduced contrast rescales the
coupling like overlap does; a misalignment between the peaks of $P$ and of
$\chi$ enters the frustration directly and is discussed at
\eqref{eq:Gamma-sign}.
\end{assumption}

\begin{assumption}[Coupling domain vs.\ aggregation domain]\label{ass:domain}
The $N$ jobs share one congestible power/cooling domain with a common cap $\Pcap$,
giving all-to-all (mean-field) throttling coupling. Justification and its
limit: co-location behind a shared cap is realistic at rack/PDU scale, but a campus
is a hierarchy of caps, so the throttling coupling is in truth block-local --
strong within a domain, weak across domains. Two domains the literature often
conflates must be separated: the coupling domain, where a shared cap creates
the interaction (possibly local), and the aggregation domain, the point of
common coupling where all loads sum and the grid sees them (necessarily global).
For the standard attractive Kuramoto family, mean-field is, at given mean coupling
strength, the topology most favourable to synchronization onset
\citep{dorfler2013,rodrigues2016,arenas2008}; for the present frustrated, delayed,
frequency-correlated coupling we adopt it as the tractable worst-case
candidate, not a proven bound -- block-local structure generically raises
the onset by diluting the coupling, but modular topologies can also support local
cluster resonances that a mean-field analysis does not see. The topology question
is flagged open in Section~\ref{sec:discussion}. Linear modal analyses of inter-site load
correlation take that correlation as exogenous \citep{modal2026}; we instead ask
whether the nonlinear dynamics generate it.
\end{assumption}

\begin{assumption}[Clock decoupling from grid frequency]\label{ass:clock}
The step rate of each job is insensitive, to first order, to grid frequency.
Justification: accelerator clocks sit behind rectification, DC regulation, and
on-board phase-locked loops, so the swing-equation entrainment \eqref{eq:swing} acts
on the grid nodes, not on the loads. Important distinction: this does
not assert insensitivity to voltage. A voltage sag does slow
computation, but through the power-delivery channel modelled by $h$ in
\eqref{eq:throttle} \citep{powerelec2025} -- i.e.\ as part of the load-dependent
coupling, not as a direct frequency entrainment. Frequency cannot couple the jobs;
the power channel can.
\end{assumption}

\begin{assumption}[Throttling response and the binding regime]\label{ass:throttle}
$h\ge0$ is nondecreasing with $h(x)=0$ for $x\le0$ and is $C^{2}$ on a neighbourhood
of an operating overshoot $s_0=\bar P_{\mathrm{tot}}-\Pcap>0$ at which $h'(s_0)>0$
-- i.e.\ the cap binds an $O(1)$ fraction of the time (tight
oversubscription). Critical caveat: if the cap binds only at peaks ($s_0\le0$
on average) then $h'(s_0)=0$ and the smooth coupling vanishes to first order; the true
interaction is then rectified -- active only while $\Ptot>\Pcap$ -- an
event-driven coupling whose period-averaged sign matches \eqref{eq:Gamma-sign} but
whose prefactor is weighted by the exceedance probability $\Pr[\Ptot>\Pcap]$ and is
hence much smaller. The theory addresses the binding regime; the rarely-binding case
is weaker (safer) and flagged for the numerical study. Gentle-throttle regime:
the reduction also treats throttling as a small perturbation, $\eta\,h\ll1$, so that
the waveforms $P_j,\chi_i$ entering the coupling are the unthrottled ones to
leading order; the analysable window is thus $s_0>0$ and $\eta h\ll1$ (soft but
frequently-binding caps), strong throttling ($\eta h\sim1$) deforming the waveforms
and being left to the numerics. Amplitude-channel omission:
throttling lowers a job's power as well as its speed, but we model only the
speed (phase) channel. The simultaneous power reduction is a negative feedback
that we expect to damp the aggregate when it is high -- heuristically
conservative -- but the expectation is not a proof: an amplitude channel carries
its own lag and could also shift the effective frustration, so its inclusion is
an open extension, not a bounded correction. Controller uniformity:
\eqref{eq:throttle} applies one scalar signal to all jobs, modulated only by
$\eta_i\chi_i(\Phi_i)$. Real power-management stacks can be hierarchical and
policy-dependent -- per-device or per-node limits, dynamic redistribution of
headroom from telemetry \citep{nvidia2026dps,dynamo2016} -- making the true coupling
asymmetric, local, and time-varying; the scalar mean-field signal is the
tractable baseline, and the hierarchy question is flagged open
(Section~\ref{sec:discussion}). Channels: the cap
acts through a fast electrical loop and slower sensing/averaging/thermal stages
with distinct lags (Assumption~\ref{ass:slowphase}; Appendix~\ref{app:grid}
classifies which stages can and cannot accumulate the sign-inverting phase).
\end{assumption}

\begin{assumption}[Weak coupling, slow phases, and the scaling of $K$]\label{ass:slowphase}
The reduction \eqref{eq:reduced} assumes (i) phases drift slowly over one control
delay, $\Phi_j(t-\tau)\approx\Phi_j(t)-\omega_j\tau$, valid for $\omega_j\tau$ not too
large, and (ii) weak coupling: the rotating-wave and retarded-argument errors are
both $O(\varepsilon)$ in the small parameter $\varepsilon:=K/\bar\omega$
(Appendix~\ref{app:reduction}), so the reduction requires $\varepsilon\ll1$;
numerically it tracks the full model up to $\varepsilon\approx0.6$ and fails by
$\varepsilon\approx1.2$ (\S\ref{sub:num-regime}). Caveat on $N$: the bare constant
$K=\tfrac12N\bar\omega\eta h'\chi_1 c_1$ is extensive unless the shared cap
saturates so that the marginal per-job throttle $h'$ scales as $1/N$. A fixed total
cap plausibly saturates in this way, but the $1/N$ scaling is a property of the
controller law, not a theorem: with per-job cap stiffness held fixed, $K$ grows
linearly in $N$ and job count alone can cross $\Kc$. Statements of the form
``adding jobs does not synchronize'' are therefore conditional on cap saturation,
and the $N$-sweeps of Section~\ref{sec:numerics} hold $K$ fixed (the intensive
normalization) by construction. The most dangerous, strongly-oversubscribed
large-$N$ limit may in any case violate weak coupling, where the averaged
theory is only qualitative and the numerical study (Section~\ref{sec:numerics}) takes
over.
\end{assumption}

\noindent Throughout, $K$ is the common coupling scale of \eqref{eq:reduced},
$a_1(\tau)$ the first-harmonic coefficient of the reduced coupling
\eqref{eq:Gamma-sign}, $g$ the density of natural rates $\omega_i$ with width
$\Delta\omega$, $\tau$ the control delay, and $\Kc$ the critical coupling
\eqref{eq:kc-general}.

\subsection{The power waveform of a single training job}\label{sub:waveform}
Fix a job $i$ running a fixed model and parallelization plan. Over one
optimization step its power draw is, to a good approximation, a two-level
periodic function of an internal phase $\Phi_i\in[0,2\pi)$ that advances
through the step:
\begin{equation}\label{eq:waveform}
  P_i(\Phi_i) \;=\;
  \begin{cases}
    P_i^{\mathrm{hi}}, & 0\le \Phi_i < 2\pi\,\delta_i \quad\text{(compute phase)},\\[2pt]
    P_i^{\mathrm{lo}}, & 2\pi\,\delta_i \le \Phi_i < 2\pi \quad\text{(communication phase)},
  \end{cases}
\end{equation}
where $\delta_i\in(0,1)$ is the compute duty cycle and $P_i^{\mathrm{hi}}>
P_i^{\mathrm{lo}}$. The fundamental period is the iteration time $T_i$, giving a
natural angular rate $\omega_i = 2\pi/T_i$. The three numbers
$(P_i^{\mathrm{hi}},P_i^{\mathrm{lo}},\delta_i)$ and the period $T_i$ are
estimable in order of magnitude from public information, which is what
makes a synthetic study feasible without facility access:
\begin{itemize}[leftmargin=1.4em,itemsep=2pt]
  \item $P_i^{\mathrm{hi}}$ from accelerator thermal design power and count,
  derated by a measured compute-phase utilization factor;
  \item $P_i^{\mathrm{lo}}$ from idle/communication-phase draw;
  \item $\delta_i$ and $T_i$ from a compute/communication time model: compute time
  from FLOPs and sustained throughput, communication time from the collective
  (all-reduce for data/tensor parallelism, all-to-all for expert parallelism)
  under an $\alpha$--$\beta$ cost model and the message volume implied by the
  parallelization plan \citep{megascale2024,lee2025overlap}.
\end{itemize}
These estimates set scales, not thresholds: a thermal design power is not
a workload power trace (instantaneous draw, averaged draw, requested and enforced
limits are distinct quantities), and the reconstruction inherits the uncertainty
of the utilization factor, the overlap fraction, and the communication model.
Dimensioning an actual site against the criteria of
Section~\ref{sec:analysis} requires measured $(P,\chi)$ waveforms, not these
public-specification surrogates (\S\ref{sec:discussion}).
It is convenient to keep the Fourier representation of \eqref{eq:waveform},
\begin{equation}\label{eq:waveform-fourier}
  P_i(\Phi) \;=\; \bar P_i \;+\; \sum_{m\ge1} c_{i,m}\cos\!\bigl(m\Phi-m\Phi^{\star}_{i,m}\bigr),
  \qquad
  \bar P_i = P_i^{\mathrm{lo}} + (P_i^{\mathrm{hi}}-P_i^{\mathrm{lo}})\,\delta_i,
\end{equation}
with $c_{i,m}=\dfrac{2(P_i^{\mathrm{hi}}-P_i^{\mathrm{lo}})}{\pi m}\,
\bigl|\sin(\pi m\delta_i)\bigr|$ for the square waveform (Appendix~\ref{app:fourier}).
The slow $1/m$ decay of $c_{i,m}$ is the quantitative statement that the load is
harmonic-rich, and is why the sinusoidal Kuramoto model is insufficient here
(Figure~\ref{fig:wave}).

\begin{figure}[t]\centering
\includegraphics[width=.62\linewidth]{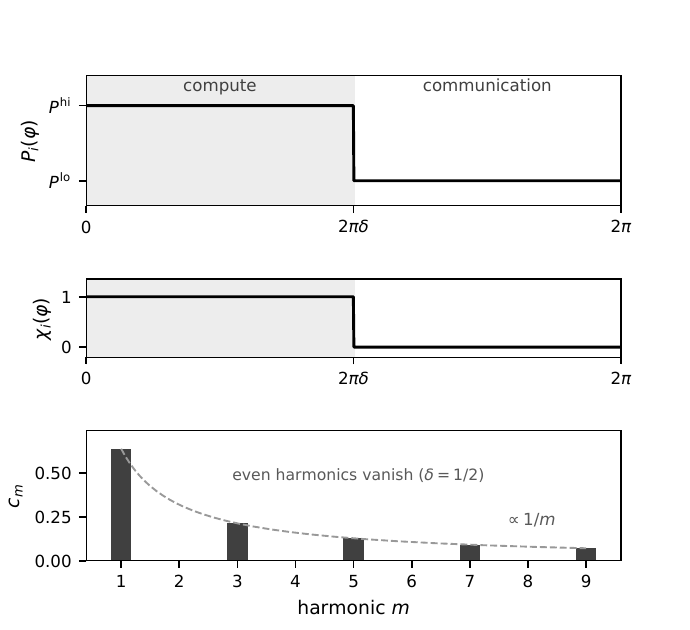}
\caption{The building blocks of the model, at the duty cycle $\delta=\tfrac12$
of the numerical study. Top: the two-level power waveform \eqref{eq:waveform}
over one iteration. Middle: the throttling susceptibility
\eqref{eq:susceptibility}, sharing the same duty cycle and reference phase --
the throttle bites only while the job is compute-bound. Bottom: the Fourier
amplitudes $c_m$ of \eqref{eq:waveform-fourier}, decaying as $1/m$ (dashed
envelope; even harmonics vanish at $\delta=\tfrac12$): the load is
harmonic-rich, which is why the coupling sign must be tracked per harmonic
(Proposition~\ref{cl:clusters}).}
\label{fig:wave}
\end{figure}

\subsection{The coupling mechanism}\label{sub:coupling}
% >>> C'est ici la correction de la faille. Texte volontairement explicite.
\paragraph{What does not couple the jobs.}
It is tempting to posit that the grid frequency entrains the jobs, by analogy with
\eqref{eq:swing}. It does not. Accelerator clocks are set by on-board phase-locked
loops downstream of rectification and voltage regulation, so first-order variations
of line frequency do not change the rate at which a step is executed, and a
model coupling training jobs through grid frequency would be physically vacuous.
Voltage is different but not a counterexample: a sag does slow computation, yet only
by reducing deliverable power and triggering throttling -- it acts through the
power channel below, not as a direct entrainment (Assumption~\ref{ass:clock}). The
coupling, if any, must therefore act on the quantity that sets the step rate: the
compute throughput.

\paragraph{What does couple them: load-dependent throttling.}
Three shared, congestible resources slow computation precisely when aggregate
demand is high.
\begin{enumerate}[leftmargin=1.6em,itemsep=2pt]
  \item \textbf{Power oversubscription with enforced caps.} Facilities provision
  power below the sum of nameplate ratings \citep{fan2007} and enforce a cap $\Pcap$
  by power management (frequency/voltage scaling, power-limit registers
  \citep{rapl2010}). When the
  instantaneous aggregate $\Ptot=\sum_j P_j$ approaches $\Pcap$, accelerators are
  throttled, lengthening the current step.
  \item \textbf{Voltage droop} on shared distribution, which reduces deliverable
  power under high aggregate current and triggers the same protective throttling.
  \item \textbf{Shared cooling}, whose finite heat-removal rate raises die
  temperatures under high aggregate dissipation, inducing thermal throttling on a
  slower timescale.
\end{enumerate}
All three slow computation when aggregate demand is high. Crucially, they slow a
given job only to the extent that it is currently compute-bound: throttling
lowers the clock at which floating-point work is done, whereas during the
communication phase a job advances at the pace of the network, not of the clock,
and is nearly insensitive to the cap. The throttle therefore acts through a
susceptibility $\chi_i(\Phi_i)$ -- a phase response curve -- that gates the
interaction by the job's own internal phase (Figure~\ref{fig:wave}, middle):
\begin{equation}\label{eq:susceptibility}
  \chi_i(\Phi_i)=
  \begin{cases}
    1, & 0\le\Phi_i<2\pi\delta_i \quad(\text{compute: throttle bites}),\\[2pt]
    0, & 2\pi\delta_i\le\Phi_i<2\pi \quad(\text{communication: throttle idle}).
  \end{cases}
\end{equation}
This gate is the physical origin of a genuine phase interaction, and it is easy to
miss: without it the throttle would multiply every job's speed by the same
aggregate-dependent factor and produce no relative-phase force -- a common
drive, not a coupling (Appendix~\ref{app:reduction}). With it, the rate of phase
advance of job $i$ depends on every other job's phase through job $i$'s own:
\begin{equation}\label{eq:throttle}
  \dot\Phi_i \;=\; \omega_i\,\Bigl[\,1 \;-\; \eta_i\,\chi_i(\Phi_i)\,
  h\!\bigl(\Ptot(t-\tau)-\Pcap\bigr)\Bigr],
  \qquad
  \Ptot(t)=\sum_{j=1}^{N} P_j\!\bigl(\Phi_j(t)\bigr),
\end{equation}
with $h$ nondecreasing, $h(x)=0$ for $x\le 0$ (no throttling below the cap),
$\eta_i\ge 0$ the dimensionless throttling sensitivity ($h$ carrying the inverse
power units), and $\tau\ge 0$ the control delay
of the power-management loop. The gate $\chi_i$ and the power waveform $P_i$ share
the same duty cycle $\delta_i$ and reference phase, since both are governed by the
compute/communication split -- an alignment assumption about the workload,
used in the reduction and relaxed at \eqref{eq:Gamma-sign}. Equations
\eqref{eq:waveform}--\eqref{eq:throttle} form a closed dynamical system: each job
drives the aggregate through $P_j(\Phi_j)$, and the aggregate feeds back on every
job's phase velocity through $\chi_i\,h$ -- the closed loop of
Figure~\ref{fig:mech}. This is the coupling the frequency
analogy was missing.

\subsection{Reduction to a generalized Kuramoto system}\label{sub:reduction}
% >>> RÉDIGÉ. Dérivation -> forme (8). C'est le pivot mathématique du papier.
The closed system \eqref{eq:waveform}--\eqref{eq:throttle} reduces, in the
weak-coupling, slow-phase regime of Assumption~\ref{ass:slowphase}, to a
generalized Kuramoto model. We give the reduction in three steps and defer the
bookkeeping to Appendix~\ref{app:reduction}.

\paragraph{Step 1 -- linearize the throttle.}
Let $s_0=\bar P_{\mathrm{tot}}-\Pcap$ be the operating overshoot
(Assumption~\ref{ass:throttle}) and $\Delta\Ptot=\Ptot-\bar P_{\mathrm{tot}}$ the
aggregate fluctuation. Then $h(\Ptot-\Pcap)=h(s_0)+h'\Delta\Ptot+O(\Delta\Ptot^2)$
with $h'=h'(s_0)>0$. The constant $h(s_0)$ multiplies
$\omega_i\eta_i\chi_i(\Phi_i)$ and, after phase-averaging $\chi_i$ to its mean
$\delta_i$, only renormalizes the natural rate to
$\omega_i^{\mathrm{eff}}=\omega_i\,[1-\eta_i\delta_i h(s_0)]$; the fluctuating part
is
\begin{equation}\label{eq:linearize}
  \dot\Phi_i \;=\; \omega_i^{\mathrm{eff}}
  \;-\; \omega_i\,\eta_i\,h'\,\chi_i(\Phi_i)\!
  \sum_{j=1}^N\Bigl[P_j\!\bigl(\Phi_j(t-\tau)\bigr)-\bar P_j\Bigr].
\end{equation}

\paragraph{Step 2 -- fundamental harmonic and the delay.}
Substitute the Fourier waveform \eqref{eq:waveform-fourier} and the matching
fundamental of the susceptibility,
$\chi_i(\Phi_i)=\delta_i+\chi_{i,1}\cos(\Phi_i-\Phi^\star_{i,1})+\cdots$ with
$\chi_{i,1}=c_{i,1}/(P_i^{\mathrm{hi}}-P_i^{\mathrm{lo}})$ and the same
reference phase $\Phi^\star_{i,1}$ as the power waveform
(Appendix~\ref{app:fourier}). For slowly drifting phases the delay becomes a phase
lag, $\Phi_j(t-\tau)\approx\Phi_j(t)-\omega_j\tau$ (Assumption~\ref{ass:slowphase}).

\paragraph{Step 3 -- average.}
The product $\chi_i(\Phi_i)\,P_j(\Phi_j(t-\tau))$ contains a fast term oscillating
at the sum frequency $\Phi_i+\Phi_j$, which averages to zero over an iteration, and
a slow term in the phase difference. Retaining the slow term (rotating-wave
averaging), the reference phases $\Phi^\star_{i,1}$ cancel and the dynamics close
into the Sakaguchi--Daido form \eqref{eq:generalized}. For a homogeneous fleet
($\omega_j\approx\bar\omega$, common $\eta,\delta,\chi_1$),
\begin{equation}\label{eq:reduced}
  \dot\Phi_i \;=\; \omega_i^{\mathrm{eff}}
  \;+\; \frac{K}{N}\sum_{j=1}^N \sin\!\bigl(\Phi_j-\Phi_i-\alpha\bigr),
  \qquad
  K \;=\; \tfrac12\,N\,\bar\omega\,\eta\,h'\,\chi_{1}\,c_{1} \;>\;0,
\end{equation}
where $c_{1}$ is the fundamental amplitude of the power waveform
\eqref{eq:waveform-fourier}, and the Sakaguchi phase
frustration is set entirely by the control delay,
\begin{equation}\label{eq:Gamma-sign}
  \alpha \;=\; \bar\omega\,\tau+\tfrac{\pi}{2},
  \qquad\Longrightarrow\qquad
  a_1(\tau)\;\equiv\;\cos\alpha\;=\;-\sin(\bar\omega\tau),
\end{equation}
$a_1(\tau)$ being the synchronizing (odd, $\sin$) coefficient of the coupling in the
convention of \eqref{eq:generalized}. The reference phase $\Phi^\star$ has dropped
out: only the delay-in-iterations $\bar\omega\tau$ survives. Equation
\eqref{eq:Gamma-sign} is the pivot of the paper, because the sign of $a_1(\tau)$
decides between coherence and incoherence. The cancellation of $\Phi^\star$ uses
the alignment of $P$ and $\chi$ (Assumption~\ref{ass:waveform}): if the
susceptibility fundamental leads the power fundamental by a phase $\Delta\varphi$,
the same computation gives $a_1=-\sin(\bar\omega\tau-\Delta\varphi)$ -- the sign
law keeps its form but the critical delay shifts by $\Delta\varphi/\bar\omega$.
Measuring $\Delta\varphi$ (the lag between the power peak and the window of cap
sensitivity) is therefore part of the two-job experiment of
Section~\ref{sec:discussion}, and the $\tau^\star=T/2$ figure below assumes
$\Delta\varphi=0$. We stress that the predicted coherence is genuine
mutual synchronization, not entrainment to an external signal. Two
statements make this precise. The averaged dynamics \eqref{eq:reduced} are
invariant under the global rotation $\Phi_i\mapsto\Phi_i+c$, so a coherent state
breaks a continuous symmetry spontaneously. The full system
\eqref{eq:throttle} does not share that rotation exactly -- the waveforms carry
reference phases, and the invariance holds only to the order of the averaging --
but it is autonomous: it contains no external clock, so the collective
phase of any locked solution is selected by the initial conditions through
time-translation symmetry, not by an outside reference. In both readings the
susceptibility $\chi_i$ (which makes the drive
depend on each job's own phase) is the indispensable ingredient: without it the
throttle is a common drive that couples nothing (Appendix~\ref{app:reduction}).

\begin{proposition}[Reduced coupling and the sign of the interaction]
\label{prop:coupling}
Under Assumptions~\ref{ass:waveform}--\ref{ass:slowphase}, a fleet of narrow rate
spread obeys \eqref{eq:reduced} with frustration $\alpha=\bar\omega\tau+\pi/2$ and
synchronizing coefficient $a_1(\tau)=-\sin(\bar\omega\tau)$. Hence the interaction
is reactive (marginal, $a_1=0$) at $\bar\omega\tau\in\{0,\pi\}$,
repulsive ($a_1<0$) for $0<\bar\omega\tau<\pi$, and attractive
($a_1>0$) for $\pi<\bar\omega\tau<2\pi$, with the first sign reversal at the
critical control delay
\begin{equation}\label{eq:taustar}
  \tau^\star \;=\; \frac{\pi}{\bar\omega} \;=\; \frac{T}{2},
\end{equation}
one half of an iteration period. Proof in Appendix~\ref{app:reduction}.
\end{proposition}

\begin{figure}[t]\centering
\includegraphics[width=.7\linewidth]{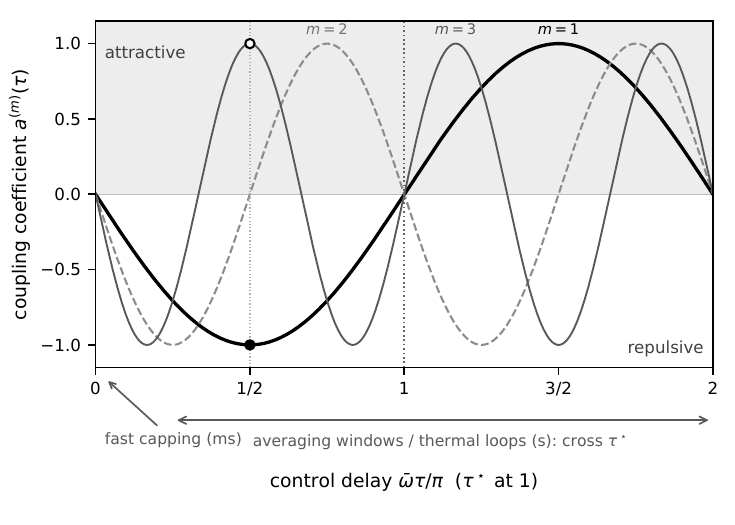}
\caption{The per-harmonic sign law $a^{(m)}(\tau)=-\sin(m\bar\omega\tau)$ for
$m=1,2,3$ (Propositions~\ref{prop:coupling} and~\ref{cl:clusters}); positive
values are attractive (locking), negative repulsive (splaying). The
fundamental is repulsive across the whole window $0<\bar\omega\tau<\pi$ and
reverses at $\tau^\star=T/2$ (right dotted line), but no delay is repulsive at
every mode: at $\bar\omega\tau=\pi/2$ (left dotted line) the fundamental is
maximally repulsive (filled dot) while the third harmonic is maximally
attractive (open dot) -- the window exploited by the harmonic lock of
\S\ref{sub:num-clusters}. Below the axis, the physical delay scales: fast
electrical capping sits at $\bar\omega\tau\approx0$, deep in the repulsive
window; accumulated dead time (polling and enforcement intervals, cascaded slow
stages) can reach and cross $\tau^\star$, whereas a single analog lag stage
cannot (Appendix~\ref{app:grid}; \S\ref{sec:discussion}).}
\label{fig:signs}
\end{figure}

\paragraph{Are realistic delays large enough? A transfer-function check.}
The sign turns on one dimensionless number -- the total phase lag of the control
loop at the iteration frequency, written $\bar\omega\tau$ -- so the operative
question is whether real loops accumulate half a cycle of lag
(Figure~\ref{fig:signs}). The word delay must here be used with care,
because different loop elements accumulate phase in qualitatively different ways
(Appendix~\ref{app:grid} develops the classification):
\begin{itemize}[leftmargin=1.4em,itemsep=2pt]
  \item \textbf{Dead time} -- transport, computation, sample-and-hold, polling
  and enforcement intervals, actuation latency -- has transfer $e^{-\ii\omega\tau}$:
  phase $\omega\tau$ grows without bound at unit gain. It is the only element
  type whose phase can cross $\pi$ on its own.
  \item \textbf{Analog lag} -- a first-order sensing or thermal stage
  $1/(1+\ii\omega\tau_{\mathrm{th}})$ -- contributes a phase
  $\arctan(\omega\tau_{\mathrm{th}})<\pi/2$ however slow the stage, and a
  damped second-order stage at most $\pi$, approached only asymptotically; its
  gain simultaneously rolls off, weakening the coupling it carries. A slow
  thermal loop is therefore not automatically past $\tau^\star$: a single
  lag stage can never invert the sign by itself (Appendix~\ref{app:grid}), and a
  thermal channel becomes dangerous only in cascade -- three or more lag
  stages (die, cold plate, coolant loop, \dots) or lag combined with dead time --
  when the total loop phase exceeds $\pi$.
  \item \textbf{Averaging windows} -- a length-$W$ moving average has transfer
  $e^{-\ii\omega W/2}\,\mathrm{sinc}(\omega W/2)$: it contributes the phase of a
  dead time $W/2$, but with an attenuation that nulls the fundamental
  exactly at $W=T$. A pure averaging window is thus self-limiting -- pushing its
  phase toward $\pi$ requires $W\ge T$, where its gain at $\bar\omega$ has already
  collapsed -- and the dangerous configuration is a moderate window
  combined with polling or actuation dead time, which rotates phase
  without paying the sinc penalty.
\end{itemize}
(A collective-communication stall, by contrast, lengthens the iteration itself:
it changes $T$, the duty cycle, and the waveform -- not the sensing-to-actuation
lag $\tau$ -- and is not a coupling-delay channel.) With iteration periods
$T\sim2$--$6$\,s \citep{megascale2024} and fast electrical capping --
short-window power limiting at the millisecond scale, model-predictive
controllers in a few milliseconds -- one has $\bar\omega\tau\ll\pi$, so the fast
loop is robustly repulsive at the fundamental (though weakly so as
$\tau\to0$, where the repulsion strength $|a_1|=\sin\bar\omega\tau$ itself
vanishes; \S\ref{sec:discussion}), and safely so at every accessible harmonic:
the first attractive mode is $m\gtrsim\pi/(\bar\omega\tau)\gg1$, whose coupling
weight decays as $1/m^2$ and whose threshold \eqref{eq:kcm} grows as $m^2$. The
attractive window is reached only through accumulated dead time in slow
channels: long polling and enforcement intervals, staged controllers, cascaded
thermal stages. For lags beyond one period ($\bar\omega\tau>2\pi$) the sign law
still applies formally through $\bar\omega\tau \bmod 2\pi$, but the
heterogeneity dephasing across the fleet grows as $\gamma\tau$ (the
$e^{\gamma\tau}$ threshold inflation of Appendix~\ref{app:oa}), so multi-wrap
predictions should be read qualitatively.
Two consequences follow, and we record them as the honest reading of
Proposition~\ref{prop:coupling}. First, the dangerous regime is a property of
dead-time-dominated control, not of throttling as such -- a design
statement, not a doom statement. Second, because $a_1$ depends on $\bar\omega\tau$, a fleet with
a broad spread of iteration periods averages $-\sin(\omega_j\tau)$ over $j$ and
washes out coherent attraction, reinforcing Corollary~\ref{cor:hetero}. Under fast
electrical capping the model therefore predicts no fleet-wide
synchronization at any accessible mode; this is the falsifiable, and reassuring,
baseline. The mode-selectivity bites at intermediate lags -- a loop at
$\bar\omega\tau\approx\pi/2$ maximizes fundamental repulsion
while handing the third harmonic its maximal attraction
(Proposition~\ref{cl:clusters}; Figure~\ref{fig:signs}).

The three readings below are the qualitative core of the mechanism. The first two
are immediate corollaries of Proposition~\ref{prop:coupling}; the third extends the
same computation harmonic by harmonic and, derived at the same level of rigour
(Appendix~\ref{app:reduction}), is stated as a proposition.
\begin{corollary}[Default repulsion at short delay -- at the fundamental]\label{cl:repulsion}
For a fast controller, $0<\bar\omega\tau<\pi$ (delay below half an iteration), the
fundamental coupling is repulsive, $a_1(\tau)<0$: jobs that compute together are
throttled together and drift apart. By Proposition~\ref{prop:kc}(i) the reduction
\eqref{eq:reduced} then admits no macroscopically coherent state; the population is
driven away from fundamental coherence ($r_1\to0$) and the aggregate does not grow
beyond the incoherent baseline of \S\ref{sub:scaling}. Two qualifications bound the
statement: it protects the fundamental mode only (Proposition~\ref{cl:clusters}), and
it holds away from a boundary layer below $\tau^\star$ whose width grows with the rate
spread, where the fast tail of $g$ is already individually attractive
(\S\ref{sub:num-sign}).
\end{corollary}
\begin{corollary}[Delay-induced attraction]\label{cl:attraction}
When the control delay exceeds $\tau^\star=T/2$, the coupling turns attractive,
$a_1(\tau)>0$, and above the critical coupling $\Kc$ of \eqref{eq:kc-general} the
fleet phase-locks ($r\to1$), producing the coherent, $N$-scaling swing the
forced-oscillation literature fears. The danger is created not by throttling but by
lagging throttling.
\end{corollary}
\noindent For identical jobs the repulsive state is a perfect splay with $r=0$,
suppressing the aggregate entirely below the incoherent baseline; at the finite rate
spread of the numerics it sits flat in $N$ and $16$--$35\%$ below the uncoupled
control -- a partial splay suppression (\S\ref{sub:num-amp}).
\begin{proposition}[Harmonic locking: the sign law is per-harmonic]\label{cl:clusters}
The waveform is harmonic-rich, and each harmonic pair carries its own frustration:
the $m$-th term of the reduced coupling is
$\propto\chi_m c_m\sin\!\bigl(m(\Phi_j-\Phi_i)-\alpha_m\bigr)$ with
$\alpha_m=m\bar\omega\tau+\tfrac\pi2$, hence a synchronizing coefficient
$a^{(m)}\propto-\sin(m\bar\omega\tau)$ (Appendix~\ref{app:reduction}). No delay is
repulsive for every $m$: a $\tau$ that protects the fundamental leaves some higher
harmonic attractive (at $\bar\omega\tau=\pi/2$, maximal fundamental repulsion,
$m=3$ is maximally attractive). Mode-$m$ self-consistency gives the lock
threshold
\begin{equation}\label{eq:kcm}
  \Kc^{(m)}=\frac{\chi_1c_1}{\chi_m c_m}\,\frac{2\gamma}{|\sin(m\bar\omega\tau)|}
  \;=\;\frac{2\gamma\,m^2}{|\sin(m\bar\omega\tau)|}
  \quad(\delta=\tfrac12,\ m\text{ odd}),
\end{equation}
with $\gamma$ the half-width of the (Lorentzian) rate density used throughout the
numerics, $g(0)=1/(\pi\gamma)$ (for a general unimodal $g$, replace $2\gamma$ by
$2/[\pi g(0)]$ as in \eqref{eq:kc-general}), so a sufficiently uniform fleet locks
into an $m$-cluster state -- $r_1$ at the
floor, $r_m=O(1)$, aggregate swing coherent at $m\bar\omega$ and scaling as $N$ --
inside the fundamental-repulsive window. Derivation in
Appendix~\ref{app:reduction}.
\end{proposition}
\noindent The numerical study confirms it quantitatively
(\S\ref{sub:num-clusters}: $r_3=0.90$ at $\gamma=0.005$, $\tau=T/4$, no saturation
involved, threshold consistent with \eqref{eq:kcm}). A saturating cap injects
additional harmonics on top of the waveform's own and can enrich these states
further.
These results correct the naive expectation that co-located jobs simply ``add
up.'' Whether a real fleet sits in the fundamental-repulsive regime (safe at $m=1$,
still exposed at higher $m$ if uniform) or the attractive regime (dangerous) is
fixed by the dimensionless group $(\bar\omega\tau,\;\eta h',\;\delta,\;
\Delta\omega/\bar\omega)$, which Section~\ref{sec:analysis} turns into
quantitative thresholds, mode by mode.

%==============================================================================
\section{Analysis}\label{sec:analysis}
% >>> RÉDIGÉ. Math indépendante de ta direction ; énoncés + esquisses, preuves
%     complètes en annexe (laissées en TODO pour que tu les valides).

\subsection{The incoherent baseline and the coherence amplification factor}
\label{sub:scaling}
Suppose first that the jobs are uncoupled ($K=0$) and their phases are
statistically independent and uniform. The aggregate is a sum of $N$ independent
zero-mean periodic signals, so its fluctuation has standard deviation
\begin{equation}\label{eq:sqrtN}
  \sigma\!\left[\Ptot-\bar P_{\mathrm{tot}}\right]
  \;=\; \Bigl(\sum_{j} \tfrac12\textstyle\sum_{m} c_{j,m}^2\Bigr)^{1/2}
  \;\sim\; \sqrt{N}\,,
\end{equation}
the familiar incoherent scaling. If instead the fundamental components phase-lock
with order parameter $r$, the coherent part of the swing -- at the fundamental
frequency, the narrowband component that couples to grid resonant modes -- has
amplitude $\sim r\,N\,c_1$, so the grid-visible peak-to-peak fluctuation scales as
\begin{equation}\label{eq:amplification}
  \frac{\text{coherent amplitude}}{\text{incoherent amplitude}}
  \;\sim\; r\,\sqrt{N}.
\end{equation}
We call $r\sqrt{N}$ the coherence amplification factor. It is the single
number that decides whether collective dynamics matter, and it cuts both ways
(Figure~\ref{fig:regimes}). In
the repulsive regime (Corollary~\ref{cl:repulsion}) the stable states minimize $r$: a
uniform splay has $r=0$ exactly and finite-$N$ leaves at most $r=O(N^{-1/2})$, so
$r\sqrt N$ does not grow; the anti-phase ordering moreover renders inter-job
correlations negative and pushes the aggregate variance below the
independent baseline -- entirely so for the perfect splay of identical jobs,
partially at the finite rate spreads of the numerics, where the aggregate sits
$16$--$35\%$ below the uncoupled control (\S\ref{sub:num-amp}). In the attractive regime
(Corollary~\ref{cl:attraction}) $r=O(1)$ and the factor grows as $\sqrt N$, so at the
hyperscale $N$ of a large campus even a modest $r$ converts a benign aggregate into
a dominant one. The same bookkeeping applies mode by mode: the coherent amplitude
at harmonic $m$ is $\sim r_m N c_m$, so a fleet locked in a higher harmonic
(Proposition~\ref{cl:clusters}) exhibits the $N$-scaling at frequency $m\bar\omega$ while
$r_1$ stays at the floor -- the amplification factor must be read on the locked
mode, not only on the fundamental. The remainder of the analysis is, in effect, the
computation of $r$ as a function of the model parameters, and in particular of the
sign and magnitude of $a_1(\tau)$ from Proposition~\ref{prop:coupling}.

\begin{figure}[t]\centering
\includegraphics[width=.85\linewidth]{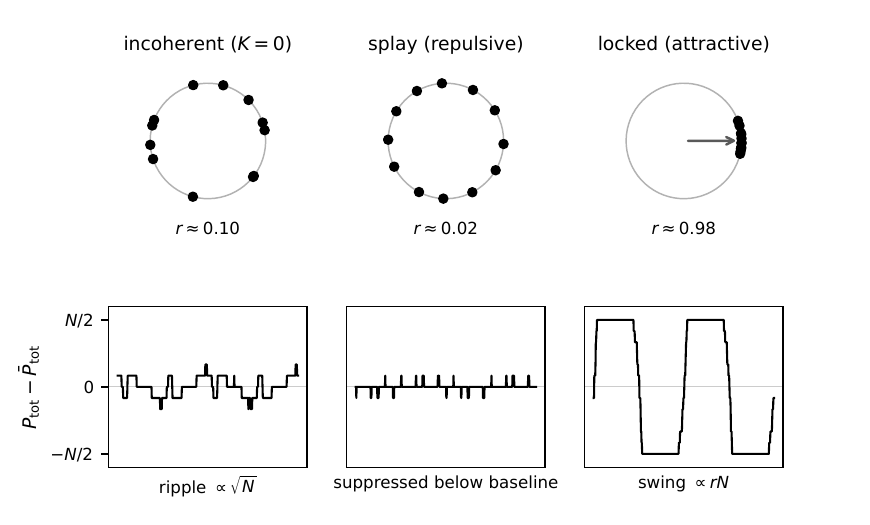}
\caption{The stake of \S\ref{sub:scaling}, schematically: $N=12$ two-level
jobs ($\delta=\tfrac12$), phase configuration (top, with the order-parameter
vector where it is macroscopic) and the exact aggregate deviation of that
configuration (bottom, common vertical scale). Independent phases leave an
incoherent ripple of scale $\sqrt N$ ($r=O(N^{-1/2})$); the splay ordering
favoured by repulsive coupling cancels the aggregate below the independent
baseline (partially so at finite rate spread, \S\ref{sub:num-amp}); the locked
state swings coherently with amplitude $\sim rN c_1$ -- the $r\sqrt N$
amplification of \eqref{eq:amplification}.}
\label{fig:regimes}
\end{figure}

\subsection{Critical coupling and the sync/splay boundary}\label{sub:critical}
% >>> Énoncé principal. La preuve complète = self-consistency généralisée.
\begin{proposition}[Onset of fleet-wide coherence]\label{prop:kc}
Let the natural rates $\omega_i$ have a unimodal, symmetric density $g\in C^2$ with
$g''(0)<0$ and width $\Delta\omega$ (a multimodal fleet of distinct job types splits into
independently-locking cohorts, one per mode), let
$a_1(\tau)=\cos\alpha=-\sin(\bar\omega\tau)$ be the synchronizing coefficient of the
coupling \eqref{eq:Gamma-sign} (Proposition~\ref{prop:coupling}), and let $K$ be the
common coupling scale of \eqref{eq:reduced}. Then:
\begin{enumerate}[label=(\roman*),leftmargin=1.9em,itemsep=2pt]
  \item if $a_1(\tau)<0$ (repulsive regime, Corollary~\ref{cl:repulsion}) the
  first-harmonic reduction \eqref{eq:reduced} admits no state with macroscopic
  fundamental coherence: $r=0$ is the only stable branch, for all $K>0$ (the full
  model can still lock in a higher harmonic, Proposition~\ref{cl:clusters} --
  a statement outside \eqref{eq:reduced});
  \item if $a_1(\tau)>0$ (attractive regime, Corollary~\ref{cl:attraction}) there is a
  critical coupling
  \begin{equation}\label{eq:kc-general}
    \Kc \;=\; \frac{2}{\pi\,a_1(\tau)\,g(0)}\,,
  \end{equation}
  above which a synchronized cluster with $r>0$ exists and is stable; $\Kc$ grows
  linearly with the heterogeneity $\Delta\omega$ of iteration rates.
\end{enumerate}
\end{proposition}
\begin{proof}[Proof sketch]
Insert the mean-field reduction \eqref{eq:meanfield}, split the population into
phase-locked ($|\Omega-\omega_i^{\mathrm{eff}}|\le Kr$) and drifting subpopulations,
and impose self-consistency $r=\langle\cos\xi_i\rangle$. The locked branch
bifurcates from $r=0$ only for $a_1(\tau)=\cos\alpha>0$ and yields
\eqref{eq:kc-general}; for $a_1(\tau)\le0$ the closure has only the $r=0$ root, with
the drifting density stationary and stable. Full computation, including the
amplitude and the inertia/delay-induced first-order transition, in
Appendix~\ref{app:critical}.
\end{proof}
\begin{corollary}[Heterogeneity protects the grid]\label{cor:hetero}
Fleets of diverse jobs (broad $g$, large $\Delta\omega$) have large $\Kc$ and are
hard to synchronize; the dangerous case is a fleet of near-identical jobs
 -- many copies of one foundation-model run, or fleets of identical fine-tunes -- 
where $\Delta\omega\to0$ and $\Kc\to0$, so that any attractive coupling locks them.
\end{corollary}

\subsection{Order of the transition: an intrinsic first-order mechanism}\label{sub:order}
The continuous bifurcation of Proposition~\ref{prop:kc} was obtained in the
narrow-spread reduction, where the source frustration $\alpha_j=\omega_j\tau+\pi/2$ of
\eqref{eq:reduced} is replaced by the common value $\alpha=\bar\omega\tau+\pi/2$. That
replacement discards the one feature that fixes the order of the transition: the
frustration is correlated with the natural frequency, with slope
$\mathrm{d}\alpha_j/\mathrm{d}\omega_j=\tau$. Folding it back, the mean field closes not
on the ordinary order parameter but on a twisted one,
\begin{equation}\label{eq:twist}
  R\,e^{\ii\Psi}=\frac1N\sum_{j} e^{\ii(\Phi_j-\alpha_j)}
  =e^{-\ii\pi/2}\,\frac1N\sum_{j} e^{\ii(\Phi_j-\omega_j\tau)},
  \qquad
  \dot\Phi_i=\omega_i^{\mathrm{eff}}+K\,R\,\sin(\Psi-\Phi_i),
\end{equation}
in which each job contributes to coherence weighted by the frequency-dependent rotation
$e^{-\ii\omega_j\tau}$ (Appendix~\ref{app:critical}). The consequence is a
self-amplifying selection. Near onset only the core $\omega_j\approx\bar\omega$ locks;
its members share almost the same rotation and add coherently, so $R$ builds faster than
linearly. As $K$ grows the locked band widens in $\omega_j$, the rotations
$e^{-\ii\omega_j\tau}$ dephase, and the marginal contribution to $R$ turns destructive.
This center-coherent / edge-destructive asymmetry is precisely the mechanism behind
explosive synchronization \citep{gomezgardenes2011,rodrigues2016}, here driven
by a frequency--frustration correlation rather than the frequency--degree correlation
of network models \citep{rodrigues2016,inertia2019}; the closest relative is the
shear diversity of \citet{montbrio2011}, where distributed nonisochronicity
likewise correlates each oscillator's phase lag with its dynamics -- there the
correlation suppresses synchrony, whereas here it has a specific physical origin and
slope, $\mathrm{d}\alpha_j/\mathrm{d}\omega_j=\tau$, and turns the onset subcritical.
Crucially it is intrinsic to the load
dynamics: it requires neither mechanical inertia nor an externally added delay reversal.

\begin{figure}[t]\centering
\includegraphics[width=.72\linewidth]{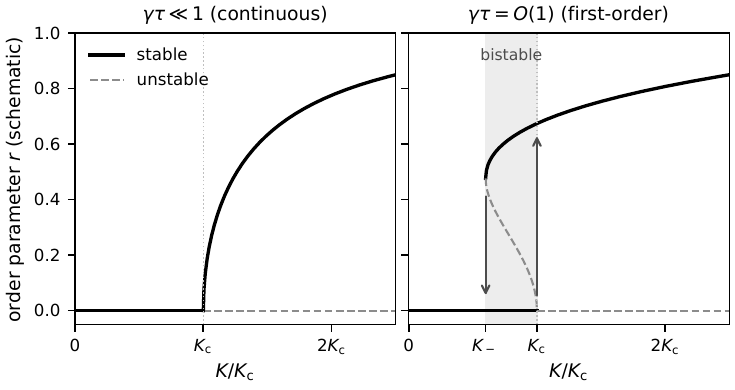}
\caption{Order of the onset -- an interpretive sketch, not a computed
bifurcation diagram (the measured continuations are in
Figure~\ref{fig:e7}). Left: for $\gamma\tau\ll1$ the bifurcation is
supercritical -- $r$ grows continuously from zero at $\Kc$. Right: for
$\gamma\tau=O(1)$ the frequency-correlated frustration \eqref{eq:twist} bends
the branch subcritical: a fold at $K_-$ (drawn at $K_-=0.6\,\Kc$, the
persistence measured in \S\ref{sub:num-order}), an unstable middle branch
(dashed, drawn from the normal-form picture -- the simulations observe only the
two stable branches), and a bistable window (shaded) with hysteretic jumps
(arrows). Inside the window a fleet tipped into the locked branch stays locked
-- the monitoring obligation of \S\ref{sec:discussion}.}
\label{fig:bif}
\end{figure}

The prediction is sharp and falsifiable: the onset, continuous for a near-identical
fleet, should turn first-order (subcritical, hysteretic) as the rate spread
$\Delta\omega$ broadens (Figure~\ref{fig:bif}). The full-model continuation of \S\ref{sub:num-order} bears
it out at the tested protocol -- the hysteresis gap sits at the finite-run floor for $\gamma\lesssim0.05$ and widens
monotonically with $\gamma$ up to the spread at which the rates one half-width
out, $\omega=\bar\omega+\gamma$, leave the attractive sign ($\gamma\approx0.35$;
the Lorentzian has no finite edge -- $\bar\omega\pm\gamma$ marks its half-mass
points -- so this is a landmark of the density, not a boundary), the gap then
saturating and receding slightly, the locked branch surviving down to a lower
turning point $K_-<\Kc$ (Table~\ref{tab:e7}). The Ott--Antonsen reduction of
Appendix~\ref{app:oa} supplies the thermodynamic-limit frame for this
observation: the onset is a Stuart--Landau bifurcation whose Landau coefficient
carries, by parity in $\tau$, a frustration-generated contribution at order
$(\gamma\tau)^2$, and the transition is first-order wherever that coefficient
turns negative. We have not evaluated the coefficient explicitly -- the cubic
closure on the OA manifold is left open -- so the analytical statement is the
structure of the normal form, and the evidence that the sign does change
is the finite-$N$ continuation itself. The interplay with
Corollary~\ref{cor:hetero} is the subtle part and must be stated plainly. Heterogeneity
raises the forward threshold $\Kc\propto\Delta\omega$ (harder to ignite from
incoherence -- the protective reading) and simultaneously opens a subcritical
lower branch $K_-<\Kc$ (a locked fleet, once tipped, stays locked over a widening range
below $\Kc$). A diverse fleet is thus harder to synchronize spontaneously but
stickier once a transient pushes it into the locked basin; the object that
governs safety is then the basin of attraction of $r=0$, not the location of $\Kc$ alone.

Two further mechanisms can deepen the bistability but are not required for it. Inertia
enters only at the grid nodes that mediate the coupling; whether their adiabatic
elimination leaves an effective second-order term for the load phases (and is the
microscopic origin of $\tau$) requires modelling them explicitly, which we do not do
here -- Appendix~\ref{app:grid} takes the first step, locating the origin of $\tau$ in the
node phase lag and showing the analog node alone is stabilizing, so that attraction
requires the total lag to exceed $\pi$. And delay-coupled first-order oscillators
independently support hysteresis \citep{yeung1999}. The intrinsic frequency--frustration mechanism,
however, operates without either.

\subsection{Engineered repulsion: phase-scattering scheduling}\label{sub:scatter}
Corollary~\ref{cl:repulsion} suggests a control lever entirely on the software
side, with two distinct ingredients whose statuses differ. Detuning
iteration rates (spreading $g$) raises every mode threshold at once --
$\Kc^{(m)}\propto\Delta\omega$ for all $m$ by \eqref{eq:kcm} -- and this
protection is independent of the control delay, since no harmonic escapes the
scaling; it is the standing, load-bearing part of the mitigation.
Fixed phase offsets (starting the fleet in a splay configuration) select
the incoherent basin, which hysteresis makes relevant (\S\ref{sub:order}); they
are not self-maintaining -- detuned rates make them drift by construction, and
holding an exact splay would require closed-loop scheduling -- so their role is
initialization and recovery, not standing protection. We state the epistemic
status plainly: phase-scattering is a mitigation implied by the model,
whose scheduler integration, offset maintenance, and interaction with
checkpointing remain to be engineered and tested; it is not a guarantee, and the
strong-coupling regime (\S\ref{sub:num-regime}) and higher-harmonic locks it
does not raise above threshold remain its failure modes to check. The cost is a
controlled loss of throughput. To leading order, forcing a
job to run at a detuned rate $\bar\omega\pm\delta\omega$ rather than at its
compute-optimal rate sacrifices a fraction $\sim\delta\omega/\bar\omega$ of its
useful work, so spreading the population over a width $\langle|\Delta\omega|\rangle$
costs a throughput penalty
\begin{equation}\label{eq:penalty}
  \frac{\Delta\text{throughput}}{\text{throughput}} \;\sim\; \frac{\langle|\Delta\omega|\rangle}{\bar\omega}
  \;\sim\; \frac{\Delta\omega}{\bar\omega},
\end{equation}
the very heterogeneity that buys safety in Corollary~\ref{cor:hetero} -- making the
safety/throughput trade-off explicit and tunable. This is an order-of-magnitude
estimate; a rigorous accounting of the offset and checkpoint overhead, which would
turn the lever into a stand-alone systems contribution, is left to future work.

%==============================================================================
\section{Related work}\label{sec:related}
Our contribution sits at the intersection of three literatures; in each we state
precisely what we borrow and what is new.

\paragraph{AI loads as grid forcing, and existing mitigations.}
A recent line treats a training facility as a single exogenous oscillation forcing
the grid, capable of exciting inter-area and sub-synchronous modes
\citep{ko2025widearea,forced2025,microsoft2025powerstab,powerelec2025}. The
proposed mitigations are of two kinds: grid-side storage \citep{hybridess2025},
and facility-side software and firmware that reshapes each job's power
profile individually -- workload insertion during communication phases,
millisecond-scale GPU power smoothing \citep{microsoft2025powerstab,oracle2025smoothing}. We take the
load model from this work but change the question: both mitigation families
treat the facility's waveform as a given object to flatten, whereas we ask
whether the waveforms of many co-located jobs can phase-lock through the
operator's own control stack. The distinction is not
cosmetic -- it is the difference between an aggregate that grows as $\sqrt N$ and one
that grows as $N$ -- and it adds a lever orthogonal to both existing families:
the relative phases of the fleet, actuated at the scheduler.

\paragraph{Spatial correlation of data-center loads.}
Closest in spirit, \citet{modal2026} analyse spatial load correlation in
data-center-dominated grids by linear modal methods, taking the correlation as
given. We instead ask whether nonlinear collective dynamics generate it, and
identify the physical coupling that would: correlation is, in our account, an emergent
order parameter, not an input.

\paragraph{Synchronization in power systems.}
The Kuramoto view of power networks casts generators as phase oscillators
coupled electromechanically through the swing equation
\citep{filatrella2008,motter2013,dorfler2013}. We invert the roles: our oscillators
are loads, they carry no mechanical inertia, and the coupling is
computational (load-dependent throttling), not electromechanical. The swing
equation reappears only at the grid nodes that mediate the interaction
(\S\ref{sub:order}).

\paragraph{Distributed-training power.}
We reconstruct the per-job waveform \eqref{eq:waveform} from characterizations of the
compute/communication cycle and its power signature
\citep{megascale2024,lee2025overlap}; the overlap of collectives with computation
\citep{lee2025overlap} is exactly the effect that bounds the harmonic content
(Assumption~\ref{ass:waveform}).

\paragraph{Generalized, delayed, and pulse-coupled synchronization.}
The machinery we use -- non-sinusoidal (Daido) coupling, Sakaguchi phase frustration,
explosive transitions, and the pulse-coupled limit -- is classical
\citep{sakaguchi1986,daido1996,mirollo1990,gomezgardenes2011,rodrigues2016,yeung1999,tanaka1997,pazo2005,hansel1993,nakao2016}, and
frequency-correlated phase lags have a direct antecedent in the distributed-shear
model of \citet{montbrio2011}. Our novelty is not the
mathematics but its derivation from a physical mechanism: the frustration
$\alpha=\bar\omega\tau+\pi/2$, the sign $a_1=-\sin\bar\omega\tau$, its per-harmonic
generalization \eqref{eq:alpham}, and the correlation slope
$\mathrm{d}\alpha_j/\mathrm{d}\omega_j=\tau$ are not posited, they fall out of the
throttle response and its delay.

%==============================================================================
\section{Numerical study}\label{sec:numerics}
We test the predictions by integrating the unaveraged delay model
\eqref{eq:waveform}--\eqref{eq:throttle} directly -- no rotating-wave averaging, no
slow-phase substitution, the delay handled as a true history -- so that the
reduction (Prop.~\ref{prop:coupling}), the threshold (Prop.~\ref{prop:kc}), and the
regimes of validity are checked rather than assumed. One substitution is retained
and stated up front: except where noted, the throttle response is the
always-binding linearization of Assumption~\ref{ass:throttle} (its exact
status, and the rectified physical response that probes the opposite regime, are
specified in the protocol below). Results are reported as tables,
whose every entry regenerates from \texttt{run.py}, and as four figures
(\texttt{figures.py}) reserved for the claims that are about shape -- the sign
reversal, the departure from the $\sqrt N$ baseline, the harmonic lock, and the
hysteresis loop -- together with a synthesizing $(\bar\omega\tau,K)$ phase diagram
(\S\ref{sub:num-phase}, \texttt{phase\_diagram.py}); the robustness of the protocol itself (time step, seeds, Fourier
truncation, rate support, rectified response) is quantified in
\S\ref{sub:num-robust}.

\paragraph{Protocol.}
The fleet is integrated with a fixed-step RK4 scheme, time step
$\Delta t=T/80$ to $T/120$ depending on the experiment (recorded per table in the
orchestrator); the control delay enters as a delay-differential equation through
a linearly interpolated history of the scalar aggregate $\Ptot(t-\tau)$,
initialized flat at its $t=0$ value. Runs last $90$--$140$ iteration periods and
all quoted statistics discard the first half as transient; halving the time step
and doubling it leave the reported observables unchanged at the quoted precision
(\S\ref{sub:num-robust}). Per-job waveforms and susceptibilities are truncated
Fourier reconstructions of the duty-$\delta$ pulse with $M$ harmonics ($M=1$
recovers the pure sinusoid; we use $M=5$ unless stated). We emphasize that the
$M$-harmonic system is the model being integrated -- a smooth member of
the regularized family of Appendix~\ref{app:reduction} -- not an imperfect
rendering of the discontinuous square wave: at $M=5$ the reconstructed pulse
overshoots to $1.09$ and undershoots to $-0.09$ (the Gibbs bounds; values per $M$
in \S\ref{sub:num-robust}), no clipping is applied to $\chi$, and the results are
stable under increasing $M$ (\S\ref{sub:num-robust}). The throttle factor
$1-\eta\chi h$ is clamped below at zero (a throttled job slows, or stalls, but
never runs backward), and the clamp-activation fraction is reported wherever it
matters (\S\ref{sub:num-regime}).

Status of the linear throttle. Unless stated, the response is linear,
$h(x)=\beta x$, with the cap referenced to the mean aggregate ($\Pcap=N\delta$),
so that the argument $x$ measures the aggregate fluctuation $\Delta\Ptot$. This
realizes the always-binding linearization of the deeply-oversubscribed
regime of Assumption~\ref{ass:throttle}: when the operating overshoot $s_0$
exceeds the fluctuation amplitude, the physical response reads
$h=h(s_0)+h'\,\Delta\Ptot$ over the whole cycle; the constant $h(s_0)$
renormalizes the rates to $\omega^{\mathrm{eff}}$ and couples nothing
(Appendix~\ref{app:reduction}), and the fluctuating part is exactly the simulated
coupling. The simulated $\omega$ is therefore $\omega^{\mathrm{eff}}$, and
excursions of the throttle factor above one are speed-ups relative to the
throttled mean rate, not above the unthrottled rate -- the physical behaviour of
an always-binding cap that momentarily relaxes. What the linear response does
not represent is the rarely-binding regime ($s_0\le0$), where the physical
interaction is rectified; that regime is probed separately with
$h(x)=\beta\max(x,0)$ (\S\ref{sub:num-clusters}, \S\ref{sub:num-robust}), and the
saturating response $h=\beta\tanh(\cdot/x_0)$ is reserved for the cluster scan
(\S\ref{sub:num-clusters}).

Rates, units, observables. Natural rates are deterministic quantiles of a
Lorentzian of half-width $\gamma$, so $g(0)=1/(\pi\gamma)$ and the
common-frustration onset is $\Kc=2\gamma/\!\cos\alpha$. The Lorentzian's tails
put a small mass at negative rates -- fraction
$\tfrac12-\tfrac1\pi\arctan(\bar\omega/\gamma)$, i.e.\ $1.6\%$ at $\gamma=0.05$,
$3.2\%$ at $0.1$, $14.8\%$ at $0.5$ -- which is unphysical for jobs (phases
running backward). We keep the untruncated family because the closed forms and
the OA reduction are exact for it, verify in \S\ref{sub:num-robust} that a
positive-truncated Lorentzian (support $(0,2\bar\omega)$, mean preserved)
reproduces the headline results, and read the broadest spreads
($\gamma\ge0.35$) as stress tests of the theory rather than fleet models. Unless
stated, $N=300$--$400$, $\delta=\tfrac12$, $M=5$. We work in normalized units
$\bar\omega=1$, $T=2\pi$, $\tau^\star=\pi$, and quote the time-averaged order
parameters $r_1,r_2,r_3$ over the second half of each run. The effective coupling
$K=\tfrac12N\bar\omega\eta h'\chi_1 c_1$ of \eqref{eq:reduced} is the swept knob. All
seeds are fixed and every table is regenerated by a single orchestrator script
(seed sensitivity in \S\ref{sub:num-robust}); raw
outputs accompany the paper (see \S\ref{sec:reproducibility}).

\paragraph{Engine validation.}
With the coupling replaced by a plain Kuramoto interaction the measured onset matches
the textbook $\Kc=2\gamma$: at $\gamma=0.05$ the order parameter is $r=0.04$ at
$K=0.4\,\Kc$, $r=0.10$ at $K=\Kc$, and $r=0.71$ at $K=2\,\Kc$ -- against the
Lorentzian closed form $r=\sqrt{1-\Kc/K}=0.71$ -- validating the
integrator and the order-parameter estimator before any model-specific run. This
check exercises the undelayed, sinusoidal limit only; the delay history, the
interpolation, and the pulse reconstruction are exercised by the convergence and
robustness checks of \S\ref{sub:num-robust} and by the agreement of the delayed
onset with the independent OA prediction \eqref{eq:oaonset}
(\S\ref{sub:num-kc}).

\subsection{Sign of the coupling and the critical delay (tests Prop.~\ref{prop:coupling})}
\label{sub:num-sign}
Sweeping the control delay at fixed coupling ($K=0.6$, $\gamma=0.05$,
Table~\ref{tab:e1}) reproduces the predicted sign reversal: the fleet is incoherent
across the repulsive window and synchronized throughout the
attractive window $\pi<\bar\omega\tau<2\pi$, with the boundary at
$\tau^\star=T/2=\pi$. Averaging over each window (excluding the marginal reactive
endpoints $|a_1|<0.15$) gives $\bar r_1=0.03$ over the repulsive window up to
$\bar\omega\tau=0.83\pi$, versus $\bar r_1=0.86$
(attractive). Two departures from the idealized picture are visible in
Figure~\ref{fig:e1} and are physics, not noise. The exactly-reactive points
$a_1\!\approx\!0$ ($\tau=0,\pi,2\pi$) are
non-generic: there the coupling is conservative and does not relax to splay, so $r_1$
is elevated -- consistent with the statement that repulsion sets in for
short but nonzero delay (Corollary~\ref{cl:repulsion}), not at $\tau=0$. And the
last nominally repulsive point, $\bar\omega\tau=0.92\pi$ ($a_1=-0.26$), reaches
$r_1=0.58$: because the frustration is heterogeneous,
$\alpha_j=\omega_j\tau+\pi/2$, the fast tail of $g$ crosses its individual
sign reversal ($\omega_j\tau>\pi$) before the mean rate does -- at this point
$17\%$ of the jobs are already attractively coupled and nucleate fleet-wide
coherence. The repulsive window is therefore eroded from above by a boundary
layer whose width scales with the rate spread (the attractive fraction is
$\tfrac12-\tfrac1\pi\arctan[(\pi/\tau-\bar\omega)/\gamma]$ for the Lorentzian
$g$): a diverse fleet loses protection slightly before $\tau^\star$, one
more face of the frequency--frustration correlation of \S\ref{sub:order}. Deep
in the window ($\bar\omega\tau\le0.83\pi$ here) the protection is intact.

\begin{table}[h]\centering\small
\begin{tabular}{ccc@{\hskip 2em}ccc}
\toprule
$\bar\omega\tau/\pi$ & $a_1$ & $r_1$ & $\bar\omega\tau/\pi$ & $a_1$ & $r_1$\\
\midrule
0.17 & $-0.50$ & 0.032 & 1.17 & $+0.50$ & 0.784\\
0.33 & $-0.87$ & 0.027 & 1.33 & $+0.87$ & 0.857\\
0.50 & $-1.00$ & 0.028 & 1.50 & $+1.00$ & 0.872\\
0.67 & $-0.87$ & 0.037 & 1.67 & $+0.87$ & 0.882\\
0.83 & $-0.50$ & 0.040 & 1.83 & $+0.50$ & 0.879\\
\bottomrule
\end{tabular}
\caption{E1 -- order parameter $r_1$ versus control delay. Repulsive window (left)
stays incoherent; attractive window (right) synchronizes; the boundary is
$\tau^\star=T/2$.}\label{tab:e1}
\end{table}

Figure~\ref{fig:e1} overlays the full delay sweep on the derived coefficient
$a_1(\tau)=-\sin\bar\omega\tau$: the measured order parameter tracks the
sign of the theoretical curve point by point, incoherent wherever $a_1<0$ and
locked wherever $a_1>0$, which is the content of Proposition~\ref{prop:coupling}
read directly off the unaveraged model.

\begin{figure}[t]\centering
\includegraphics[width=.72\linewidth]{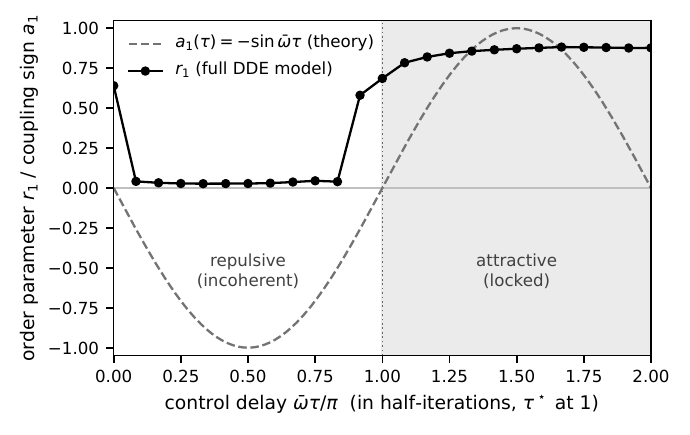}
\caption{E1 -- measured order parameter $r_1$ (dots, unaveraged DDE model) against the
derived coupling coefficient $a_1(\tau)=-\sin\bar\omega\tau$ (dashed) across one
period of the control delay. The fleet is incoherent across the repulsive
window and locked throughout the attractive window (shaded),
with the reversal at the predicted $\tau^\star=T/2$. Two elevated features are
explained in the text: the exactly-reactive points ($\tau=0,\pi$, conservative
coupling), and the boundary-layer point at $\bar\omega\tau=0.92\pi$, where the
fast tail of the rate distribution is already individually attractive and
nucleates coherence. Regenerated by \texttt{sim/figures.py} from
\texttt{e1\_sign.csv}.}
\label{fig:e1}
\end{figure}

\subsection{Critical coupling and heterogeneity (tests Prop.~\ref{prop:kc}, Cor.~\ref{cor:hetero})}
\label{sub:num-kc}
At $\tau=3\pi/2$ (pure attraction, $\cos\alpha=1$) two analytical values frame the
onset: the leading-order threshold $\Kc=2\gamma$ \eqref{eq:kc-general}, and the
exact onset of the averaged model,
$\Kc^{\mathrm{avg}}=2\gamma\,e^{\gamma\tau}$ \eqref{eq:oaonset}. These are not
competing predictions for one system but statements about two members of the
model hierarchy (Appendix~\ref{app:oa}); the full delay model is expected to sit
at the leading-order value. The coarse scan behaves accordingly: for
$\gamma=0.05$ ($\Kc=0.10$) the order parameter rises from
$r_1=0.05$ at $K=0.5\,\Kc$ to $r_1=0.16$ at $K=\Kc$ and $r_1=0.57$ at $1.5\,\Kc$,
and the crossing of the marker $r_1>0.3$ lands at
$K^\star=0.057,\,0.113,\,0.226$ for $\gamma=0.025,\,0.05,\,0.10$, against
$2\gamma=0.05,\,0.10,\,0.20$. The marker necessarily sits above the bifurcation
(it detects $r_1=0.3$, not $0^+$) and the sweep grid is coarse, so this scan
establishes the linear growth of the onset with the rate spread
(Corollary~\ref{cor:hetero}) rather than a threshold to percent accuracy. The
fine sweep of \S\ref{sub:num-robust} supplies the accuracy: at $\gamma=0.05$,
$N=1200$, steps of $0.0075$ in $K$, the measured branch fits the supercritical
closed form $r_1=\sqrt{1-\Kc/K}$ with $\Kc=0.100$--$0.101$ across
$K=0.12$--$0.15$, against $2\gamma=0.100$ and
$\Kc^{\mathrm{avg}}=0.127$: the full model's onset is the leading-order
$2\gamma$, and the $e^{\gamma\tau}$ inflation is confirmed to be a property of
the averaged model, not of the delay dynamics (Appendix~\ref{app:oa}).

\subsection{Coherence amplification (tests \S\ref{sub:scaling})}\label{sub:num-amp}
Table~\ref{tab:e3} reports the aggregate fluctuation normalized by $\sqrt N$,
together with an uncoupled control ($K=0$, same estimator and window), at
$\gamma=0.05$ and $K=0.6$ -- a rate spread above the mode-3 threshold
$\gamma^\star=K/18$ of \eqref{eq:kcm}, so that the repulsive window is genuinely
protective here (the sub-threshold behaviour is the subject of
\S\ref{sub:num-clusters}). In the repulsive regime the fluctuation is flat in $N$,
while in the attractive regime it grows roughly fourfold from $N=50$ to $N=800$
-- the $r\sqrt N$ amplification of \eqref{eq:amplification}, here with
$r\!\approx\!0.88$. The independent-phases value per job is
$\bigl(\tfrac12\sum_m c_m^2\bigr)^{1/2}=0.48$ (\eqref{eq:sqrtN} at $M=5$,
$\delta=\tfrac12$, unit amplitude; the untruncated square wave gives
$\sqrt{\delta(1-\delta)}=0.5$); the $K=0$ control measures $0.42$--$0.45$
against it, the shortfall being the finite estimation window (slow beats between
nearby rates are under-sampled). The same estimator and window are applied to
every branch, so this bias cancels to first order in branch-to-control
comparisons, and the control's shortfall ($6$--$12\%$ of the theoretical
baseline) bounds its size; against the theoretical baseline the suppression
below would read larger, so quoting it against the control is the conservative
choice. The repulsive branch sits a further
$16$--$35\%$ below the control: fast repulsive capping does not merely fail to
amplify, it actively suppresses the aggregate below independence --
the partial form of the perfect-splay cancellation of \S\ref{sub:scaling},
visible also in $r_1$ (e.g.\ $0.016$ at $N=800$, half the random-phase floor
$N^{-1/2}\approx0.035$). Figure~\ref{fig:e3} shows the three curves; at much
smaller spreads the harmonic lock takes over instead
(\S\ref{sub:num-clusters}).

\begin{table}[h]\centering\small
\begin{tabular}{ccccc}
\toprule
$N$ & uncoupled ($K=0$) & repulsive $\sigma[\Ptot]/\sqrt N$ & attractive $\sigma[\Ptot]/\sqrt N$ & attractive $r_1$\\
\midrule
50  & 0.45 & 0.29 & 2.60 & 0.88\\
200 & 0.44 & 0.36 & 5.20 & 0.87\\
800 & 0.42 & 0.35 & 10.4 & 0.87\\
\bottomrule
\end{tabular}
\caption{E3 -- the repulsive aggregate fluctuation stays flat in $N$, below the
uncoupled control; the attractive fluctuation grows as
$r\sqrt N$.}\label{tab:e3}
\end{table}

\begin{figure}[t]\centering
\includegraphics[width=.72\linewidth]{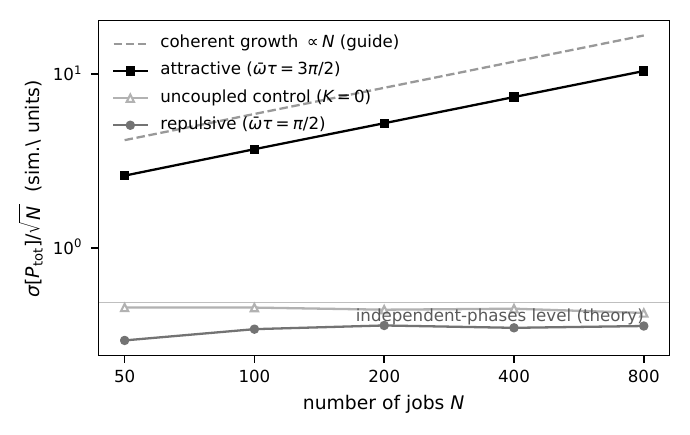}
\caption{E3 -- aggregate power fluctuation divided by $\sqrt N$, versus fleet size.
The attractive fleet follows the coherent $\propto N$ guide (dashed); the
uncoupled $K=0$ control (open triangles) sits just under the theoretical
independent-phases line (solid grey, $0.48$; finite-window effect); the repulsive
fleet is flat in $N$ and $16$--$35\%$ below the control -- active partial
splay suppression, not mere non-amplification. Log--log axes.
Regenerated by \texttt{sim/figures.py} from \texttt{e3\_amplification.csv}.}
\label{fig:e3}
\end{figure}

\subsection{Order of the transition and its heterogeneity dependence (tests \S\ref{sub:order})}\label{sub:num-order}
For a near-identical fleet ($\gamma=0.05$), sweeping $K$ up and then down at
$\tau=3\pi/2$ with warm-started phases, the up and down branches coincide: the
base-model continuation shows a maximum gap of $0.08$, located at the transition and
attributable to finite-run relaxation, and the corresponding $\gamma=0.05$ member of
the heterogeneity series shows a residual of $0.10$, sitting on the upper plateau
(Figure~\ref{fig:e7}, left) -- both at the relaxation floor, neither a loop. The
bifurcation is continuous, exactly as the cubic analysis of
Appendix~\ref{app:critical} predicts in the narrow-spread limit.

This is, however, only the narrow-fleet face of the transition. Repeating the
continuation across increasing rate spread (Table~\ref{tab:e7}) tests the intrinsic
frequency--frustration mechanism of \S\ref{sub:order}: a plain Sakaguchi model would
stay continuous at every $\gamma$, whereas the twist \eqref{eq:twist} predicts a loop
that opens as $\gamma$ grows. It does -- Figure~\ref{fig:e7} shows the two faces side
by side, the closed continuation of the near-identical fleet against the open loop of
the diverse one, with the locked branch of the latter surviving below $\Kc$. The hysteresis gap is at the relaxation floor for
$\gamma\le0.05$ and then climbs monotonically -- $0.10\to0.68$ as $\gamma$ goes
$0.05\to0.2$, a sevenfold rise -- while the warm-started locked branch survives down to
$K/\Kc\approx0.6<1$, the signature of a subcritical lower turning point $K_-<\Kc$. The
onset is therefore continuous for homogeneous fleets and first-order for
diverse ones at this protocol, with no inertia in the model: the load dynamics alone carry the
first-order seed (\S\ref{sub:order}). (At the broadest spreads, $\gamma\ge0.35$, the
oscillators beyond one half-width, $|\omega-\bar\omega|>\gamma$ -- half the
Lorentzian mass; the density has no finite edge -- are past their individual
sign reversal ($a_1(\omega\tau)<0$ at $\bar\omega+\gamma$, last
column), so fewer jobs participate and the gap saturates, then recedes slightly; the
effect is the correlation, not mere coupling strength.) Two cautions bound the
classification. First, it is operational: ``first-order'' encodes a gap well
above the relaxation floor in a finite-run, warm-started continuation at one
sweep rate; distinguishing a genuine subcritical fold from slow relaxation would
require longer runs, finer increments, sweep-rate and system-size scaling, and a
computation of the unstable branch, none of which are attempted here. Second,
the gap is sensitive to the distribution's tails: on the positive-truncated
Lorentzian of \S\ref{sub:num-robust} the $\gamma=0.2$ loop closes to near the
floor, so the wide loops of Table~\ref{tab:e7} are carried in part by the far
tail of the unbounded Lorentzian -- the family for which the OA structure is
exact, but not a fleet model. The first-order claim should accordingly be read
as established for the heavy-tailed family and open for bounded-support fleets
(\S\ref{sub:num-robust}).

\begin{table}[h]\centering\small
\begin{tabular}{ccccc}
\toprule
$\gamma$ & $a_1(\bar\omega)$ & $a_1(\bar\omega+\gamma)$ & $\max|r_1^{\uparrow}-r_1^{\downarrow}|$ & onset\\
\midrule
0.02 & $+1.00$ & $+1.00$ & 0.11 & continuous\\
0.05 & $+1.00$ & $+0.97$ & 0.10 & continuous\\
0.10 & $+1.00$ & $+0.89$ & 0.22 & first-order\\
0.20 & $+1.00$ & $+0.59$ & 0.68 & first-order\\
0.35 & $+1.00$ & $-0.08$ & 0.67 & first-order\\
0.50 & $+1.00$ & $-0.71$ & 0.57 & first-order\\
\bottomrule
\end{tabular}
\caption{E7 -- order of the synchronization onset versus rate heterogeneity, at
$\tau=3\pi/2$. The third column evaluates the coupling coefficient at the
half-width rate $\bar\omega+\gamma$ (the Lorentzian's half-mass point -- the
density has no finite edge), i.e.\ $-\sin[(\bar\omega+\gamma)\tau]$. The
hysteresis gap is at the finite-run floor ($\approx0.10$) for a
near-identical fleet and widens with $\gamma$, the locked branch persisting to
$K_-<\Kc$: consistent with the intrinsic frequency--frustration correlation
$\alpha_j=\omega_j\tau+\pi/2$ rendering the onset first-order without any
inertia. The classification is protocol-bound and tail-sensitive
(\S\ref{sub:num-robust}): on the positive-truncated Lorentzian the $\gamma=0.2$
loop closes to near the floor.}
\label{tab:e7}
\end{table}

\begin{figure}[t]\centering
\includegraphics[width=.9\linewidth]{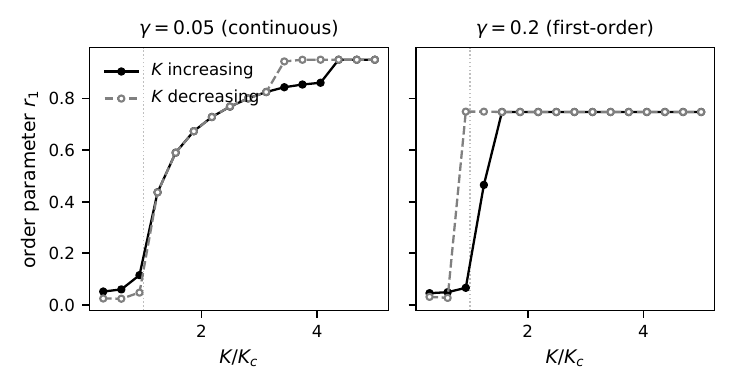}
\caption{E7 -- warm-started up/down continuation of the order parameter at
$\tau=3\pi/2$. Left: near-identical fleet ($\gamma=0.05$), branches coincide up to
the finite-run relaxation floor (residual offset $\le0.10$, on the upper plateau,
not at the onset): continuous. Right: diverse fleet ($\gamma=0.20$), the loop opens
(maximum gap $0.68$) and the locked branch persists below the forward threshold
(dotted line at $K/\Kc=1$) -- the subcritical, first-order onset driven by the
frequency-correlated frustration, with no inertia in the model. Identical
parameters and seed as Table~\ref{tab:e7}; regenerated by
\texttt{sim/figures.py} (curves cached in \texttt{e7\_loop\_gamma*.csv}).}
\label{fig:e7}
\end{figure}

\subsection{Regime of validity: gentle versus strong throttling (tests Asm.~\ref{ass:throttle}, \ref{ass:slowphase})}
\label{sub:num-regime}
The sign law is a gentle-throttle result. Table~\ref{tab:e5} sweeps the coupling
(hence the throttle depth) and reports, in the repulsive regime, the order parameter
and the fraction of evaluations at which the throttle would drive a job backward and
is clamped. While throttling is gentle ($K\le0.6$, no clamping) the repulsive state
stays incoherent ($r_1\approx0.02$, the sign law holds). Once the throttle leaves the
gentle regime ($K\ge1.2$) the sign law fails through two distinct mechanisms,
and the table separates them: at $K=1.2$ and $4.8$ the clamp -- a hard nonlinearity
that can lock even the repulsive regime -- fires massively (clamp fraction $0.96$ and
$0.99$), while at $K=2.4$ the clamp never fires (fraction $0.00$) yet $r_1=0.47$: there
the coupling is simply too strong for the averaged reduction
($K/\bar\omega=2.4$, Assumption~\ref{ass:slowphase}), and the smooth dynamics are
already multistable. Strong capping is thus its own hazard whether or not it
saturates. This is the strong-throttle/strong-coupling regime flagged in
Assumptions~\ref{ass:throttle}--\ref{ass:slowphase} as beyond the averaged theory; the
numerics, not the leading-order sign law, govern there.

\begin{table}[h]\centering\small
\begin{tabular}{ccc}
\toprule
$K$ & repulsive $r_1$ & clamp fraction\\
\midrule
0.3 & 0.025 & 0.00\\
0.6 & 0.024 & 0.00\\
1.2 & 0.648 & 0.96\\
2.4 & 0.473 & 0.00\\
4.8 & 0.564 & 0.99\\
\bottomrule
\end{tabular}
\caption{E5 -- the repulsive sign law holds in the weak-coupling, gentle-throttle
regime ($K\le0.6$) and fails beyond it, whether through hard saturation ($K=1.2$,
$4.8$) or through strong smooth coupling alone ($K=2.4$, clamp never fires).}\label{tab:e5}
\end{table}

\subsection{Harmonic locking in the repulsive window (tests Proposition~\ref{cl:clusters})}
\label{sub:num-clusters}
Proposition~\ref{cl:clusters} predicts that the fundamental-repulsive window is not safe for
a uniform fleet: at $\bar\omega\tau=\pi/2$ ($\tau=T/4$, maximal fundamental repulsion)
the third harmonic is maximally attractive, with lock threshold
$\Kc^{(3)}=18\gamma$ from \eqref{eq:kcm} -- at the working coupling $K=0.6$, a rate
spread below $\gamma^\star=K/18\approx0.033$ should lock. Table~\ref{tab:e8} sweeps
$\gamma$ at fixed $K=0.6$, $N=800$, with the linear throttle; the clamp never fires
(clamp fraction $0.000$ throughout), so no saturation is involved. The prediction is
quantitative: for $\gamma\le0.02$ the fleet locks into a three-cluster state --
$r_1$ stays at the incoherent floor while $r_3$ reaches $0.95$ and the aggregate
swing rises nearly an order of magnitude above the independent-phases level,
coherent at
$3\bar\omega$ (a macroscopic $r_3$ triggers the amplification of
\S\ref{sub:scaling} read at mode $m=3$: the mode-3 coherent component alone
predicts $\sigma[\Ptot]/\sqrt N\approx r_3\sqrt N\,c_3/\sqrt2=4.0$ at $N=800$,
against $4.05$ measured -- the swing scales as $N$) -- while for
$\gamma\ge0.04$ the lock is extinguished, bracketing the predicted $\gamma^\star$;
the intermediate point $\gamma=0.03$ is partially locked ($r_3=0.30$).
Figure~\ref{fig:e8} shows the crossover. The lock survives the rectified,
throttle-only response $h(x)=\beta\max(x,0)$: at $\gamma=0.005$ it gives
$\sigma[\Ptot]/\sqrt N=3.09$ with $r_3=0.79$, so the state is not an artifact of the
linear surrogate's speed-up branch. One reading caution: this configuration is by
design the sharpest test of \eqref{eq:kcm} -- $\delta=\tfrac12$ removes the
even harmonics and $\tau=T/4$ maximizes $a^{(3)}$ -- so the experiment establishes
existence and threshold of the harmonic escape, not its typicality. Realistic
duty cycles, finite ramps (which temper $c_m$, Appendix~\ref{app:fourier}),
waveform asymmetries, multi-pulse structure (Assumption~\ref{ass:waveform}), and
drifting periods all reshape the mode weights, and their systematic exploration
is left open.

\begin{table}[h]\centering\small
\begin{tabular}{ccccc}
\toprule
$\gamma$ & $\sigma[\Ptot]/\sqrt N$ & $r_1$ & $r_3$ & $m=3$ locked (pred.\ $\gamma<0.033$)\\
\midrule
0.002 & 4.05 & 0.045 & 0.95 & yes\\
0.005 & 3.84 & 0.042 & 0.90 & yes\\
0.01  & 3.49 & 0.041 & 0.82 & yes\\
0.02  & 2.63 & 0.032 & 0.61 & yes\\
0.03  & 1.30 & 0.019 & 0.30 & marginal\\
0.04  & 0.36 & 0.016 & 0.06 & no\\
0.05  & 0.35 & 0.016 & 0.05 & no\\
\bottomrule
\end{tabular}
\caption{E8 -- harmonic locking at maximal fundamental repulsion
($\bar\omega\tau=\pi/2$, $K=0.6$, $N=800$, linear throttle, clamp inactive). Below
the predicted threshold $\gamma^\star=K/18\approx0.033$ the fleet locks at the third
harmonic: $r_1$ at the floor, $r_3=O(1)$, aggregate nearly an order of magnitude
above the
independent-phases level ($\approx0.48$, \S\ref{sub:num-amp}) and coherent at
$3\bar\omega$.}\label{tab:e8}
\end{table}

\begin{figure}[t]\centering
\includegraphics[width=.9\linewidth]{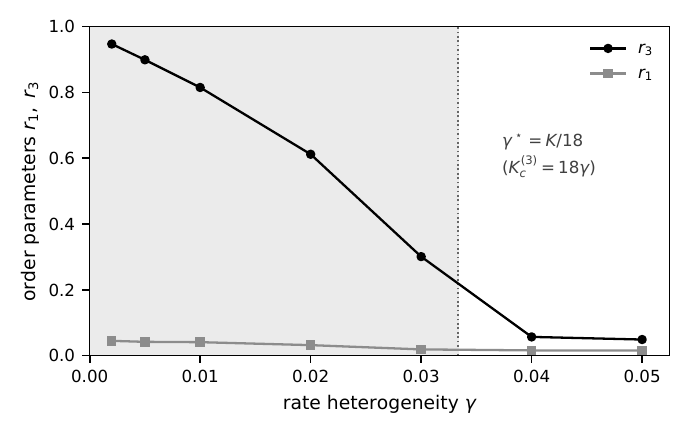}
\caption{E8 -- the crossover of Table~\ref{tab:e8}: $r_3$ (black) collapses as the
rate spread crosses the mode-3 threshold $\gamma^\star=K/18$ (dotted line, shaded
region locked) predicted by \eqref{eq:kcm}, while $r_1$ (grey) never leaves the
incoherent floor -- a fundamental-only monitor is blind to the transition.
Regenerated from \texttt{e8\_harmonic.csv} by \texttt{sim/figures.py}.}
\label{fig:e8}
\end{figure}

An earlier scan (E6) had looked for the two-cluster signature -- low $r_1$
with elevated $r_2$ -- under a saturating cap at $\delta=0.3$, and found none (the
largest $r_2-r_1$ was $+0.02$). Equation~\eqref{eq:kcm} explains the null: at
$\delta=0.3$ the mode-3 coupling weight is
$\chi_3c_3/(\chi_1c_1)\approx1/62$ of the fundamental, far below threshold at every
tested $\gamma$ -- that scan probed a weak harmonic in a suppressed regime, not the
absence of the phenomenon. The clean route to cluster states is the one above: odd
harmonics at $\delta=\tfrac12$, no saturation required.

\subsection{Synthesis: the delay--coupling phase diagram}\label{sub:num-phase}
The one-dimensional slices E1 (\S\ref{sub:num-sign}) and E8
(\S\ref{sub:num-clusters}) are two cuts of a single object -- the steady-state
coherence of the fleet over the $(\bar\omega\tau,K)$ plane, mapped from
incoherent initial phases (forward map, not warm-started, so no hysteresis
path dependence). Figure~\ref{fig:phasemap} shows it. Panel~(a), the
fundamental order parameter $r_1$ at $\gamma=0.05$, is
Proposition~\ref{prop:coupling} made two-dimensional: the repulsive window
$0<\bar\omega\tau<\pi$ stays incoherent at every coupling, while the
attractive window locks above the analytic onset $\Kc(\tau)=2\gamma/a_1(\tau)$
of \eqref{eq:kc-general} (overlaid), whose diverging tongue at the window edges
-- the $1/a_1$ growth as $a_1\to0^+$ -- the colour field tracks. Panel~(b), the
aggregate swing $\sigma[\Ptot]/\sqrt N$ at a near-identical fleet
$\gamma=0.005$, exposes what $r_1$ hides: a bright lock tongue centred at
$\bar\omega\tau=\pi/2$, deep inside the fundamental-repulsive window
where panel~(a) is dark, is the third-harmonic cluster state of
Proposition~\ref{cl:clusters}, its boundary following the analytic threshold
$\Kc^{(3)}(\tau)=18\gamma/|\sin 3\bar\omega\tau|$ of \eqref{eq:kcm} (overlaid).
The figure states the control message of \S\ref{sec:discussion} in one view:
fast electrical capping (the left edge, $\bar\omega\tau\to0$) is dark in both
panels, and the two bright tongues of~(b) -- the fundamental lock above
$\bar\omega\tau=\pi$ and the $m{=}3$ tongue at $\pi/2$ -- are only the two
lowest-order of a comb, the dark bands between them hosting still-higher
harmonics whose thresholds grow as $m^2$ \eqref{eq:kcm} and so sit just above
the plotted coupling. By the per-harmonic sign law no delay is repulsive at
every mode: tuning $\tau$ only relocates the locking mode, whereas rate
diversity, which lifts every threshold \eqref{eq:kcm} at once
(\S\ref{sub:scatter}), darkens all modes together.
Generated by \texttt{sim/phase\_diagram.py}.

\begin{figure}[t]\centering
\includegraphics[width=\linewidth]{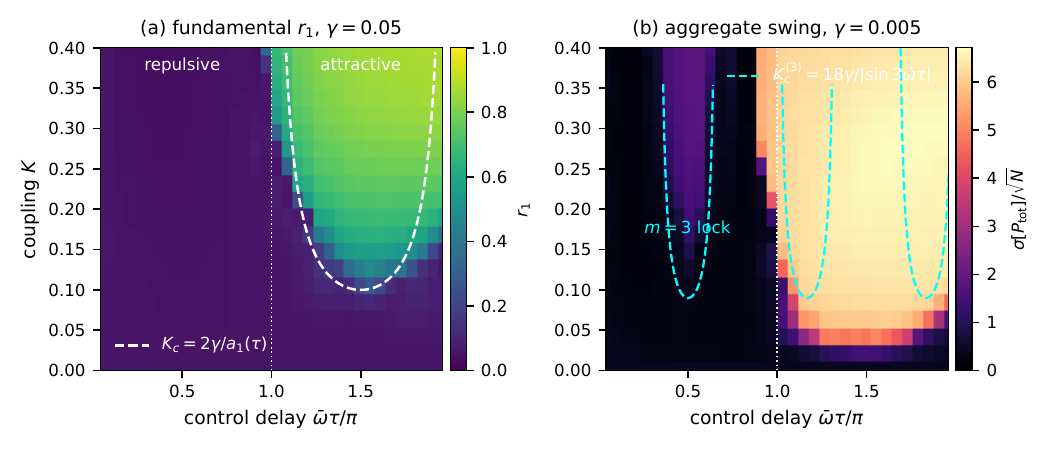}
\caption{E9 -- steady-state phase diagram over the control-delay/coupling plane,
from the full DDE model (forward map from incoherent phases). \textbf{(a)}
Fundamental coherence $r_1$ at $\gamma=0.05$: incoherent (dark) across the whole
repulsive window $\bar\omega\tau<\pi$, locked (bright) in the attractive window
above the analytic onset $\Kc=2\gamma/a_1(\tau)$ (white dashed), which diverges
at the reactive edges. \textbf{(b)} Aggregate swing $\sigma[\Ptot]/\sqrt N$ at a
near-identical fleet $\gamma=0.005$: a third-harmonic lock tongue (bright)
centred at $\bar\omega\tau=\pi/2$ sits inside the fundamental-repulsive
window where (a) is dark, following the analytic $\Kc^{(3)}=18\gamma/
|\sin3\bar\omega\tau|$ (cyan dashed) of \eqref{eq:kcm} -- a fundamental-only
monitor is blind to it. The attractive window $\bar\omega\tau>\pi$ is bright as
well, where the fundamental itself locks as in~(a). No delay is dark at every
mode; only rate diversity is.
Regenerated by \texttt{sim/phase\_diagram.py} (cached in
\texttt{phasemap\_a.csv}, \texttt{phasemap\_b.csv}).}
\label{fig:phasemap}
\end{figure}

\subsection{Robustness of the numerical protocol}\label{sub:num-robust}
Six protocol-level checks respond to the natural objections to the setup above;
all are regenerated by \texttt{sim/checks.py} and \texttt{sim/checks2.py} (raw
outputs \texttt{CHECKS.md}, \texttt{CHECKS2.md}).

Time step. Three representative configurations (repulsive $\tau=\pi/2$;
attractive $\tau=3\pi/2$; harmonic lock $\gamma=0.005$) rerun at
$\Delta t=T/80$, $T/160$, $T/320$ change $r_1$ and $r_3$ by less than
$5\times10^{-4}$ (e.g.\ attractive $r_1=0.8706/0.8701/0.8701$): the
discretization error is far below the precision quoted anywhere in the paper.

Seeds. Five independent initial-phase draws give
$r_1=0.0277\pm0.0007$ (repulsive), $r_1=0.8720\pm0.0002$ (attractive),
$r_3=0.8990\pm0.0002$ (harmonic lock): the fixed-seed tables are representative,
not cherry-picked. (Rate draws are deterministic quantiles, so seed variation
probes initial phases; a study of basin statistics under random rate draws is
not attempted.)

Fourier truncation. The reconstructed duty-$\tfrac12$ pulse has Gibbs
bounds $[-0.137,1.137]$ at $M=1$, $[-0.094,1.094]$ at $M=5$, $[-0.090,1.090]$ at
$M=15$; no clipping is applied, the $M$-harmonic system being the model by
definition. The observables are insensitive to $M$: the attractive $r_1$ moves
$0.8716\to0.8766$ and the harmonic-lock $r_3$ moves $0.8992\to0.9006$ between
$M=5$ and $M=15$ -- sub-percent drifts, no qualitative change -- so no
conclusion rides on the negative-susceptibility excursions of the truncation.

Positive-rate fleet. On the truncated Lorentzian (support
$(0,2\bar\omega)$, mean preserved, no backward-running oscillators): the sign-law
window means are unchanged ($r_1=0.024$--$0.038$ repulsive, $0.88$--$0.91$
attractive at $\gamma=0.05$); the onset marker at $\gamma=0.10$ is unchanged
($0.226$); the hysteresis, however, is reshaped -- the $\gamma=0.10$ and $0.20$
loops close to near the relaxation floor ($0.074$ and $0.058$ against $0.22$ and
$0.68$ untruncated) while the $\gamma=0.35$ loop reopens strongly ($0.737$, the
locked branch persisting to $K/\Kc\approx0.61$). Sign law and forward threshold
are therefore tail-insensitive; the order of the transition at
intermediate spreads is not, which is why \S\ref{sub:num-order} states the
first-order classification with its distribution dependence attached.

Rectified (physical) response. At gentle depth ($K=0.6$) the rectified
throttle $h=\beta\max(x,0)$ reproduces the sign law -- $r_1=0.031$--$0.037$
across the repulsive window (clamp fraction $0$), $0.83$--$0.86$ across the
attractive window -- and sustains the harmonic lock (\S\ref{sub:num-clusters}).
At $K=1.2$ the rectified response saturates the clamp ($98\%$ of evaluations)
and lands in the strong-throttle regime of \S\ref{sub:num-regime}, where no
smooth-theory statement applies -- consistent with the linear-throttle behaviour
there.

Fine onset. At $\gamma=0.05$, $N=1200$, two seeds, steps of $0.0075$ in
$K$: $r_1(K)$ rises from the finite-size floor between $K=0.0975$ and $0.1125$
and fits the supercritical closed form $r_1=\sqrt{1-\Kc/K}$ with
$\Kc=0.100$--$0.101$ over $K=0.12$--$0.15$; used in \S\ref{sub:num-kc} to
separate the leading-order threshold from the averaged-model value.

\subsection{Reproducibility}\label{sec:reproducibility}
The engine (full delay-differential model, RK4, Fourier-reconstructed waveforms with
the clamped throttle), a single orchestrator that regenerates every table above
from fixed seeds, and the robustness-check scripts accompany the paper as
ancillary files (\texttt{sim/kdc.py}, \texttt{sim/run.py},
\texttt{sim/checks.py}, \texttt{sim/checks2.py}, \texttt{sim/figures.py},
\texttt{sim/phase\_diagram.py} (the E9 map),
with the raw outputs under \texttt{sim/outputs/}; NumPy, total runtime
$\approx25$ minutes on a laptop). The reported quantities are
deterministic functions of the listed parameters; re-running the orchestrator
reproduces Tables~\ref{tab:e1}--\ref{tab:e8} and every
in-text number exactly, and a companion script regenerates
Figures~\ref{fig:e1}--\ref{fig:e8} from the same outputs (the continuation curves of
Figure~\ref{fig:e7} are re-simulated with the identical parameters and seed, and the
reproduced hysteresis gaps match Table~\ref{tab:e7} to the reported precision),
while \texttt{sim/phase\_diagram.py} regenerates the phase map
(Figure~\ref{fig:phasemap}) from its own cached sweep. The
five schematic figures (Figures~\ref{fig:mech}--\ref{fig:bif}) are drawn from
closed forms by \texttt{sim/figures\_schematic.py}, with no simulation data. No
proprietary or facility data are used: the study is entirely synthetic, calibrated on
public accelerator specifications and analytical collective-communication models, as
described in \S\ref{sub:waveform}.

%==============================================================================
\section{Discussion}\label{sec:discussion}

\paragraph{Epistemic status.}
Everything in this section is model-implied guidance. The study is
synthetic end to end -- no facility trace, no hardware test, no scheduler
experiment (\S\ref{sec:reproducibility}) -- so the statements below are
hypotheses ranked by a theory whose sign prediction is directly testable, not
validated operating rules; the two-job measurement under Falsifiability
is the gate to pass before treating any of them as one.

\paragraph{Dimensioning: model-implied guidance for power-capping controllers.}
The sign $a_1=-\sin\bar\omega\tau$ reads directly as design guidance for power-capping
controllers. A fast loop ($\bar\omega\tau\ll\pi$) is repulsive and, in the
weak-coupling regime of the theory, cannot ignite fundamental coherence; the
danger appears only when the total loop lag approaches half an
iteration -- for a heterogeneous fleet, slightly before it, the fast rate
tail crossing its individual sign reversal first (\S\ref{sub:num-sign}) -- which
in practice means accumulated dead time: long polling and enforcement
intervals, staged controllers, cascaded thermal stages. A single analog or
first-order thermal lag cannot invert the sign however slow it is, and a pure
averaging window nulls the fundamental before its phase reaches $\pi$
(\S\ref{sub:reduction}; Appendix~\ref{app:grid}). The actionable
prescription is therefore to keep power capping in the fast electrical loop and to
avoid letting a dead-time-dominated or cascaded-slow controller become the binding one. A finer
reading of the same law tempers the word ``stabilizing,'' however: the strength
of the repulsion is $|a_1|=\sin\bar\omega\tau$, which vanishes as $\tau\to0$ (the
instantaneous limit is marginally reactive, $a_1\to0^-$: a pure frequency shift
that neither locks nor actively spreads the fleet) and peaks at $\bar\omega\tau=\pi/2$,
i.e.\ $\tau=T/4$. But the delay that maximizes fundamental repulsion is also the delay
that maximizes the attraction of the third harmonic
($a^{(3)}=-\sin(3\bar\omega\tau)=+1$ at $\bar\omega\tau=\pi/2$,
Proposition~\ref{cl:clusters}, confirmed in \S\ref{sub:num-clusters}): no control delay is
repulsive for every mode, so tuning $\tau$ selects which mode is dangerous
rather than making the fleet safe. A fast loop removes fundamental
attraction within the model; only rate diversity above the mode thresholds \eqref{eq:kcm} protects all
harmonics at once -- the $\tau$-lever is mode-selective, the job-mix lever is global.
In physical time, with the measured $T\approx2$--$6$\,s, the fundamental danger window
$\tau>T/2$ is $1$--$3$\,s: sub-millisecond electrical capping sits deep in the
repulsive window; second-scale polling and enforcement intervals straddle
$\tau^\star$ (a second-scale averaging window contributes only half its
length in phase and attenuates the fundamental it feeds back,
\S\ref{sub:reduction}); and multi-stage thermal channels can accumulate several
turns of phase, where the sign law applies only qualitatively
(\S\ref{sub:reduction})
\citep{ko2025widearea,microsoft2025powerstab}. Equally, a
controller that hard-saturates the cap leaves the gentle regime entirely
(\S\ref{sub:num-regime}), where the clean sign law no longer protects: aggressive
clamping is its own hazard. We note, finally, what separates this guidance from a
dimensioned rule: all simulations run in normalized units, and evaluating
$K=\tfrac12N\bar\omega\eta h'\chi_1c_1$ on a site requires four measured
quantities -- the waveform contrast ($c_1$), the gate fundamental ($\chi_1$), the
cap stiffness $h'$ (throttle response per watt of aggregate excess), and the
per-job sensitivity $\eta$ -- all accessible from telemetry during a controlled
capping experiment, and none reliably reconstructible from public specifications
(\S\ref{sub:waveform}). Until they are measured, the criterion is dimensionless
guidance, not a sizing rule.

\paragraph{Job mix: heterogeneity is protective, but double-edged.}
Corollary~\ref{cor:hetero}, confirmed numerically (\S\ref{sub:num-kc}), says diversity
of iteration rates raises the forward threshold $\Kc$; the fleet easiest to
ignite from incoherence is the homogeneous one -- many replicas of a single
foundation-model run, or large fleets of identical fine-tunes. But the same
frequency-correlated frustration that, once broadened, turns the onset first-order
(\S\ref{sub:order},~\S\ref{sub:num-order}) also opens a subcritical lower branch
$K_-<\Kc$: a diverse fleet is harder to synchronize spontaneously yet stickier
once a transient tips it into the locked basin (with the caveat of
\S\ref{sub:num-robust}: the stickiness at intermediate spreads is established on
the heavy-tailed rate family, and on bounded-support fleets the measured loop
closes there before reopening at large spreads). The harmonic escape sharpens the
stakes from the other side: heterogeneity is not only what raises the fundamental
threshold, it is the only lever that protects every mode at once, all the
thresholds \eqref{eq:kcm} scaling with $\Delta\omega$. Heterogeneity is therefore a
real asset against spontaneous onset but not a guarantee under perturbation -- which
sharpens, rather than weakens, the case for phase-scattering scheduling
(\S\ref{sub:scatter}) as a software-only, zero-storage mitigation, whose throughput cost
we estimated and whose rigorous accounting we leave open.

\paragraph{Detection: an early-warning signal in existing telemetry.}
The subcritical branch turns the order parameter from a theoretical device into an
operational quantity. Power traces are already collected at rack and PDU
level; where jobs map one-to-one onto nodes or racks -- the common case for
large training partitions -- these traces are attributable per job, and because
the waveform is near two-level (\S\ref{sub:waveform}), each job's
instantaneous phase $\Phi_j$ can be read off its own trace by thresholding, with
$r_m=|\langle e^{\ii m\Phi_j}\rangle|$ evaluated online at low cost. Where
telemetry aggregates several jobs, or a job spans domains, per-job phase
extraction is a disaggregation problem we do not solve here. In the
benign regime $r_1$ stays at (or below) its incoherent scale $O(N^{-1/2})$ -- and
because the cluster states of Proposition~\ref{cl:clusters} are invisible to $r_1$
(\S\ref{sub:num-clusters}: $r_1\approx0.04$ while the aggregate scales as $N$), the
watch must extend to the first few harmonics $r_m$, or more simply to the aggregate
variance against its incoherent level. A persistent excursion
after a disturbance -- a cap-saturation episode, a co-scheduled restart wave, a
thermal event that lengthens the effective delay -- signals entry into a locked
basin, which hysteresis makes sticky (\S\ref{sub:num-order}): waiting for the
disturbance to pass does not guarantee return to incoherence once $K>K_-$. The
operational reading is a watch-and-release policy: monitor $r_1,r_2,r_3$ against a
threshold calibrated by the Rayleigh statistic of circular uniformity -- under
incoherence $N r_m^2$ is asymptotically exponential with unit mean, so the threshold
$r_m>\sqrt{\ln(1/p)/N}$ bounds the per-window false-alarm rate by $p$ (the effect of
finite estimation windows and of waveform correlation on this baseline is left
open: the statistic assumes independent uniform phases, so this is a calibration
baseline, not a validated detector) -- and on a persistent excursion apply
a transient rate detuning (\S\ref{sub:scatter}) to eject the fleet from the locked
branch, rather than waiting for it to relax on its own.

\paragraph{Limits, honestly.}
Five assumptions carry the result and bound it. (i) Mean-field, all-to-all coupling:
a real campus is a hierarchy of power domains, so the true coupling is
block-local; for attractive Kuramoto families mean-field is the most
synchronization-prone topology and block dilution raises $\Kc$, but for the
present frustrated, delayed coupling this conservatism is plausible rather than
proven (Assumption~\ref{ass:domain}). (ii) Leading harmonic: justified
by the $O(1/m^2)$ decay (Appendix~\ref{app:reduction}) except under saturation
or multi-pulse waveforms (Assumption~\ref{ass:waveform}). (iii)
Gentle throttling: the strong-throttle regime is multistable and beyond the averaged
theory (\S\ref{sub:num-regime}). (iv) Order of the onset: the loads carry no mechanical
inertia, yet the transition is not simply continuous -- the reduction's
frequency-correlated frustration makes it first-order and hysteretic for diverse fleets
(\S\ref{sub:order}, Table~\ref{tab:e7}). Our evidence for the order is numerical
(finite $N$, warm-started continuation, and tail-sensitive:
\S\ref{sub:num-robust}); the explicit Landau coefficient
(Appendix~\ref{app:oa}), a finite-size-scaling study, and the separate
question of whether grid-node inertia adds an effective second-order term (which may be
the very origin of $\tau$), remain open. (v) Determinism: the model's rates are
constant and noiseless, whereas real iteration times drift and jump --
stragglers, checkpoints, data loading, expert routing, failures and restarts
\citep{megascale2024}; phase noise generically raises coherence thresholds and
can wash out weakly locked states as well as unlock hysteretic ones, so the
deterministic fleet is the cleanest-case laboratory, and robustness of every
locked state reported here to realistic rate noise is untested.

\paragraph{Falsifiability.}
Two predictions are directly testable. The coupling has a measurable sign:
two co-capped jobs at short control delay should anti-correlate (repulsion), a
controlled experiment that would confirm or refute the mechanism -- and that
simultaneously calibrates the power--susceptibility alignment $\Delta\varphi$ of
\eqref{eq:Gamma-sign} and the loop's actual phase budget
(Appendix~\ref{app:grid}). And the route to
danger is through loop lag or rate homogeneity, not through job count:
adding jobs at fixed heterogeneity does not, by itself, synchronize a fleet
behind a fixed, saturating total cap -- at fixed per-job cap stiffness the
coupling is extensive and job count alone can cross the threshold
(Assumption~\ref{ass:slowphase}), so the experiment must control the cap policy.

\paragraph{Beyond the fence line: resonance with grid modes.}
The analysis above stays inside the facility, where the operator's stakes and
levers are; what a locked fleet does to the grid is the natural extension. Because
the grid responds in narrowband near its electromechanical modes
(\S\ref{sub:inertia}), the resonance-relevant quantity is the coherent
fundamental, amplified by $r\sqrt N$ (\S\ref{sub:scaling}); a synchronized
fleet whose iteration rate lands on a grid mode is the worst case, and detuning away
from it is a second, frequency-domain, use of the same scheduling lever. The
frequency ranges do overlap: with $T\sim2$--$6$\,s the iteration fundamental sits
at $f=1/T\approx0.17$--$0.5$\,Hz, inside the $0.1$--$1$\,Hz inter-area band
\citep{ko2025widearea}, so the hazard cannot be dismissed by scale separation
alone. Whether a given site would actually excite a given mode is a different
question: it depends on the modal frequency and damping, the residue at the
site's electrical location, and the amplitude visible at the point of common
coupling -- none of which are modelled here, and no resonance claim is made. The
overlap is what motivates the companion, network-model work; rate-detuning
remains the zero-storage lever if that work confirms the exposure.

%==============================================================================
\section{Conclusion}\label{sec:conclusion}
We asked whether co-located AI training jobs, usually modelled as independent grid
forcings, can instead synchronize -- turning an aggregate that should grow as
$\sqrt N$ into one that grows as $N$. Answering it required identifying the physical
coupling, which is neither grid frequency (the clocks are decoupled) nor a common
forcing (which cannot break symmetry), but load-dependent throttling gated by
each job's own compute phase. Linearizing that mechanism yields a generalized
Kuramoto model whose Sakaguchi frustration $\alpha=\bar\omega\tau+\pi/2$ and
synchronizing coefficient $a_1=-\sin\bar\omega\tau$ are derived, not posited,
giving a sign that flips from repulsive to attractive at a control delay of half an
iteration, $\tau^\star=T/2$, and a critical coupling $\Kc=2/[\pi g(0)a_1]$ that
collapses for homogeneous fleets. The same reduction carries a frequency-correlated
frustration $\alpha_j=\omega_j\tau+\pi/2$ whose twisted order parameter drives the onset
first-order for diverse fleets without any inertia -- an intrinsic
mechanism of the load dynamics, established here numerically. Direct
simulation of the unaveraged delay model reproduces each prediction within its
stated regime -- the sign reversal at
$\tau^\star$, the onset at the leading-order threshold $2\gamma/\cos\alpha$
(resolved to percent accuracy by a fine sweep, which also separates it from the
averaged model's $e^{\gamma\tau}$-inflated value, Appendix~\ref{app:oa}) and its
linear growth with rate
spread, the $r\sqrt N$ amplification, the harmonic lock at its sharpest point (a
near-identical fleet reaches $r_3\approx0.9$ at the very delay that maximizes
fundamental repulsion, with a threshold consistent with \eqref{eq:kcm}), and a
transition that is continuous for
near-identical fleets but opens a hysteresis loop as the heavy-tailed fleet
diversifies (with the protocol and tail-sensitivity caveats of
\S\ref{sub:num-order}) --
and maps the boundaries of the theory: strong coupling, saturating or not, is a
separate multistable regime beyond the averaged sign law.

The model's headline for the operator is reassuring and actionable: fast electrical
power capping stabilizes the fundamental -- it makes the leading coupling repulsive
and holds the aggregate below the independent-phases baseline -- while the coherent swing the
forced-oscillation literature fears requires accumulated dead time in the
control loop (a single analog or thermal lag stage cannot invert the sign,
Appendix~\ref{app:grid}), a hard-saturating cap, or a fleet uniform enough to lock at the
fundamental or at some higher harmonic. Because the sign law is per-harmonic, no
control delay is safe for every mode: the one protection that covers them all is
rate diversity, which no cap tuning can substitute for. Every ingredient
of the dangerous regime is thus a site-configuration choice, not an external
circumstance; the residual risk left by hysteresis and by the $r_1$-invisible
cluster states is legible in telemetry the site already collects, through the order
parameter and its harmonics; and the remedy lives in the
scheduler, through deliberate heterogeneity and phase scattering. Unlike the
single-waveform problem addressed by storage and by per-job power smoothing, the
collective form of the problem is, in the model, governable with levers
the operator already holds -- and whether it is so in the machine room is
precisely what the two-job experiment below would decide.

Three directions remain open. First, the first-order onset is here established
numerically (finite $N$, heavy-tailed rate family); the explicit Landau
coefficient of the Ott--Antonsen normal form \eqref{eq:landau-b}, a self-consistent
treatment of the twisted order parameter
\eqref{eq:twist} -- whose locked band is asymmetric in $\omega$, so the standard closure
does not apply -- a finite-size-scaling study, and the bounded-support question
raised in \S\ref{sub:num-robust} would together pin the lower turning point
$K_-$, the thermodynamic limit, and the distribution dependence. The companion question, whether the second-order grid
node of Appendix~\ref{app:grid} bestows an additional effective load inertia when
carried to higher order, is now separable from the order of the transition rather than a
prerequisite for it. Second, the realistic rectified, strong-throttle,
and block-local-topology regimes deserve their own analysis, the present mean-field,
gentle-throttle theory being only their conservative envelope. Third, the sign of the
coupling is experimentally accessible: a two-job, co-capped measurement would test the
mechanism at its root.

%==============================================================================
\appendix
\section{Fourier coefficients of the two-level waveform}\label{app:fourier}
% >>> RÉDIGÉ (calcul direct).
For the square waveform \eqref{eq:waveform} with duty cycle $\delta$ and amplitude
$A=P^{\mathrm{hi}}-P^{\mathrm{lo}}$, the mean is $\bar P=P^{\mathrm{lo}}+A\delta$ and
the Fourier amplitudes follow from the one-sided complex coefficient
$\hat c_m=\frac{A}{2\pi}\int_0^{2\pi\delta} e^{-\ii m\Phi}\,\dd\Phi
=\frac{A}{\pi m}\,\sin(\pi m\delta)\,e^{-\ii\pi m\delta}$, whose real cosine
amplitude is $c_m=2|\hat c_m|$:
\begin{equation}
  c_{m} \;=\; \frac{2A}{\pi m}\,\bigl|\sin(\pi m\delta)\bigr|,
  \qquad m\ge1,
\end{equation}
with reference phase $\Phi^\star_{m}=\pi\delta$, the centre of the compute
window, shifted by $\pi/m$ where $\sin(\pi m\delta)<0$ (the sign absorbed in the
modulus): every harmonic is referenced to the same pulse centre, which is what
makes the waveform and susceptibility fundamentals share $\Phi^\star_{1}$ in the
reduction (\S\ref{sub:reduction}). At $\delta=\tfrac12$ this is
the textbook square-wave series, $c_m=2A/(\pi m)$ for odd $m$. The $1/m$ envelope
confirms the
harmonic-rich character used in \S\ref{sub:generalized}. The susceptibility
$\chi$ of \eqref{eq:susceptibility} is the same indicator with unit amplitude, so it
shares the reference phase and has coefficients $\chi_m=c_m/A=2(\pi m)^{-1}
|\sin(\pi m\delta)|$; in particular $\chi_1 c_1 = A\,\chi_1^2$, which is the
amplitude product entering the coupling constant $K$ of \eqref{eq:reduced}.

\paragraph{Finite ramps (trapezoidal waveform).}
Real transitions are not instantaneous: let each edge occupy a phase width
$2\pi\rho$ with $\rho\ll\delta$. Convolving the square wave with a width-$2\pi\rho$
boxcar multiplies each coefficient by the sinc factor $\operatorname{sinc}(m\rho)
=\sin(\pi m\rho)/(\pi m\rho)$, giving
\begin{equation}\label{eq:trapezoid}
  c_m^{\mathrm{trap}} \;=\; \frac{2A}{\pi m}\,\bigl|\sin(\pi m\delta)\bigr|\,
  \operatorname{sinc}(m\rho).
\end{equation}
Ramps therefore leave the fundamental essentially unchanged ($\operatorname{sinc}
(\rho)\approx1$) and only accelerate the high-harmonic roll-off from $1/m$ to
$1/m^2$ beyond $m\sim1/\rho$. The reduction of \S\ref{sub:reduction}, which keeps
$m=1$, is thus insensitive to ramp width; only the cluster-forming harmonics of
Proposition~\ref{cl:clusters} are tempered by finite ramps.

\section{Reduction to the generalized Kuramoto form}\label{app:reduction}
This appendix proves Proposition~\ref{prop:coupling}, i.e.\ derives
\eqref{eq:reduced}--\eqref{eq:taustar} from \eqref{eq:throttle} and bounds the
discarded terms.

\paragraph{Linearization (Step 1).}
With $s_0=\bar P_{\mathrm{tot}}-\Pcap$ and $\Delta\Ptot=\Ptot-\bar P_{\mathrm{tot}}$,
Assumption~\ref{ass:throttle} gives $h=h(s_0)+h'\Delta\Ptot+O(\Delta\Ptot^2)$. The
phase average of $\chi_i$ over one step is its duty cycle,
$\langle\chi_i\rangle=\delta_i$, so the $h(s_0)$ term contributes a mean slowdown
$\omega_i\eta_i\delta_i h(s_0)$ and oscillatory corrections that, lacking any
$\Phi_j$ dependence, do not couple jobs: the oscillatory part of $\chi_i h(s_0)$ is a
single-oscillator, $\Phi_i$-dependent speed modulation that merely redefines the
isochron phase and is absorbed into the phase variable. Absorbing the mean into
$\omega_i^{\mathrm{eff}}=\omega_i[1-\eta_i\delta_i h(s_0)]$ yields
\eqref{eq:linearize}.

\paragraph{The slow-phase (retarded) approximation (Step 2).}
For job $j$, $\Phi_j(t-\tau)=\Phi_j(t)-\omega_j\tau-\varepsilon_j$ with
$\varepsilon_j=\int_{t-\tau}^{t}(\dot\Phi_j-\omega_j)\,\dd s$. The integrand is the
coupling term of \eqref{eq:reduced}, of magnitude $O(K)$ and oscillatory at rate
$\bar\omega$, so its integral over a delay obeys $|\varepsilon_j|=O(K/\bar\omega)$:
the correction is controlled by weak coupling alone, and -- crucially --
stays small even in the attractive window $\bar\omega\tau\sim\pi$, where the lag
$\omega_j\tau$ itself is order one and is kept exactly, not expanded.
Replacing $\Phi_j(t-\tau)$ by $\Phi_j(t)-\omega_j\tau$ thus costs only
$O(K/\bar\omega)$, the same order as the rotating-wave step below.

\paragraph{Averaging and the sum-frequency term (Step 3).}
Insert the fundamentals
$\chi_i(\Phi_i)\supset\delta_i+\chi_1\cos(\Phi_i-\Phi^\star)$ and
$P_j(\Phi_j(t-\tau))-\bar P_j\supset c_1\cos(\Phi_j-\omega_j\tau-\Phi^\star)$ into
\eqref{eq:linearize}. The product of the two fundamentals splits by the
product-to-sum identity into
\begin{equation}\label{eq:prod2sum}
  \chi_1 c_1\cos(\Phi_i-\Phi^\star)\cos(\Phi_j-\omega_j\tau-\Phi^\star)
  =\tfrac{\chi_1 c_1}{2}\Bigl[\underbrace{\cos(\Phi_i-\Phi_j+\omega_j\tau)}_{\text{slow}}
  +\underbrace{\cos(\Phi_i+\Phi_j-\omega_j\tau-2\Phi^\star)}_{\text{fast}}\Bigr].
\end{equation}
The fast term rotates at $\dot\Phi_i+\dot\Phi_j\approx2\bar\omega$, whereas the slow
term varies at the rate of phase differences, $O(\Delta\omega)+O(K)\ll\bar\omega$;
averaging over one iteration suppresses the fast term to relative order
$O(K/\bar\omega)$, the standard rotating-wave error, small under
Assumption~\ref{ass:slowphase}. The $\delta_i$ part of $\chi_i$ multiplies the same
$P_j$ fundamental and gives a $\Phi_i$-independent common drive
$F_i(t)=-\omega_i\eta_i\delta_i h'\,\Delta\Ptot(t-\tau)$ with
$\Delta\Ptot=Nc_1\,r\cos(\psi-\Omega\tau-\Phi^\star)$, which is $O(N)$ and $O(r)$.
Its size is, however, deceptive: because $F_i$ does not depend on $\Phi_i$, it enters
a phase difference only through the heterogeneity of its prefactor,
$\dot\Phi_i-\dot\Phi_k\supset(\omega_i\eta_i\delta_i-\omega_k\eta_k\delta_k)\,G(t)$
with $G$ common. For uniform $\eta_i\delta_i$ this vanishes identically: a common
time-dependent frequency is a pure gauge (a choice of rotating frame) and cannot
create coherence, however large. For heterogeneous $\eta_i\delta_i$ it is $O(r)$ but
near onset the locked band has width $O(Kr)\to0$, so the prefactor is nearly constant
across the locked oscillators and the drive is absorbed into $(\Omega,\psi)$; its
differential, synchronization-relevant part is higher order in $r$ and does not move
$\Kc$ at leading order (it can renormalize $\Kc$ at $O(r)$ away from onset, which we
do not pursue; the argument is a near-onset one and does not cover strongly
locked branches or the broadest spreads -- there the numerical study is the
evidence, and it integrates the term in full, the simulated $\chi$ retaining its
mean $\delta$). Keeping the slow term of
\eqref{eq:prod2sum}, the reference phase $\Phi^\star$ cancels and
\begin{equation*}
  \dot\Phi_i=\omega_i^{\mathrm{eff}}
  -\omega_i\eta_i h'\tfrac{\chi_1 c_1}{2}\sum_{j}\cos(\Phi_i-\Phi_j+\omega_j\tau).
\end{equation*}
Using $-\cos\theta=\sin(\theta-\tfrac\pi2)$ and
$\cos(\Phi_i-\Phi_j+\omega_j\tau)=\cos\bigl((\Phi_j-\Phi_i)-\omega_j\tau\bigr)$ gives
$\sin\bigl((\Phi_j-\Phi_i)-\omega_j\tau-\tfrac\pi2\bigr)$, i.e.\ \eqref{eq:reduced}
with $\alpha_j=\omega_j\tau+\tfrac\pi2$ and
$K=\tfrac12N\bar\omega\eta h'\chi_1 c_1>0$. For narrow $g$, $\alpha_j\to\alpha
=\bar\omega\tau+\tfrac\pi2$, whence $a_1=\cos\alpha=-\sin(\bar\omega\tau)$ and the
sign pattern of Proposition~\ref{prop:coupling}; $a_1$ first vanishes from below at
$\bar\omega\tau=\pi$, giving $\tau^\star=\pi/\bar\omega=T/2$. For heterogeneous duty
cycles ($\delta_i\neq\delta_j$) the reference phases no longer cancel and leave a
residual frustration $\pi(\delta_j-\delta_i)$, a duty-cycle-dependent contribution to
$\alpha$ omitted in the homogeneous reduction. \hfill\qedsymbol

\paragraph{Error control (first-order averaging).}
Steps 2--3 are a first-order averaging of \eqref{eq:linearize} in the small parameter
$\varepsilon=K/\bar\omega$: writing the dynamics as
$\dot\Phi=\omega^{\mathrm{eff}}+\varepsilon\,F(\Phi,t)$ with $F$ the throttling
coupling, the reduced system \eqref{eq:reduced} is the average of $F$ over one
iteration. The two discarded contributions -- the retarded-argument correction
$\varepsilon_j=O(\varepsilon)$ of Step~2 and the sum-frequency term of Step~3 -- are
each $O(\varepsilon)$ relative to the retained coupling. Provided the waveform is
regularized to Lipschitz class -- the finite ramps of Appendix~\ref{app:fourier}, under
which $h\circ P$ and $\chi$ are $C^1$ and $F$ is smooth and bounded -- the first-order
averaging theorem \citep{sandersverhulst2007} applies: solutions of the full and
averaged systems starting $O(\varepsilon)$-close stay $O(\varepsilon)$-close on the long
time scale $t=O(1/\varepsilon)$ over which coherence develops. We therefore take
\eqref{eq:reduced} as the leading-order model and, per Assumption~\ref{ass:slowphase},
fall back at $\varepsilon\gtrsim0.6$ to direct simulation of the unreduced system
(\S\ref{sub:num-regime}). The idealized two-level waveform is the
$\varepsilon$-independent discontinuous limit of this family; the Fourier truncation
used throughout is what renders $F$ smooth mode by mode, so the averaging applies
harmonic by harmonic with the same $O(\varepsilon)$ control.

\paragraph{Harmonics beyond the fundamental: per-mode frustration and threshold.}
Repeating Step 3 for the $m$-th harmonic pair
($\chi_m\cos(m\Phi_i-m\Phi^\star)$ against $c_m\cos(m\Phi_j-m\omega_j\tau-m\Phi^\star)$),
the slow term is $\tfrac{\chi_m c_m}{2}\cos\!\bigl(m(\Phi_i-\Phi_j)+m\omega_j\tau\bigr)$
and the same $-\cos\theta=\sin(\theta-\pi/2)$ step gives the coupling
$\propto\chi_m c_m\sin\!\bigl(m(\Phi_j-\Phi_i)-\alpha_m\bigr)$ with
\begin{equation}\label{eq:alpham}
  \alpha_m=m\bar\omega\tau+\tfrac\pi2,
  \qquad a^{(m)}\propto\cos\alpha_m=-\sin(m\bar\omega\tau):
\end{equation}
the delay multiplies by $m$, the $\pi/2$ of the velocity coupling does not. Its
relative size is $\chi_m c_m/(\chi_1 c_1)=[\sin(\pi m\delta)/(m\sin(\pi\delta))]^2
=O(1/m^2)$, so the truncation to $m=1$ is controlled for the dynamics near
fundamental onset -- but each mode has its own instability. In the sector
$\Psi_i=m\Phi_i$ the $m$-th term is a standard Sakaguchi coupling with frequencies
$m\omega_i$ (spread $m\Delta\omega$) and strength $m\,K\chi_mc_m/(\chi_1c_1)$;
the self-consistency of Appendix~\ref{app:critical} then gives the mode-$m$
threshold
\begin{equation*}
  \Kc^{(m)}=\frac{\chi_1c_1}{\chi_mc_m}\cdot\frac{2\gamma}{|\sin(m\bar\omega\tau)|}
  =\frac{2\gamma m^2}{|\sin(m\bar\omega\tau)|}\quad(\delta=\tfrac12,\ m\ \text{odd}),
\end{equation*}
which is \eqref{eq:kcm}: the $m$-cluster lock of Proposition~\ref{cl:clusters}, active
whenever $-\sin(m\bar\omega\tau)>0$ and the fleet is uniform enough, saturation or
not. A hard cap additionally amplifies high harmonics beyond the waveform's own
$c_m$.

\section{Self-consistency for the critical coupling}\label{app:critical}
This appendix proves Proposition~\ref{prop:kc}, then develops (without closing) the
twisted-order-parameter mechanism that fixes the order of the transition.

\paragraph{Mean-field form.}
With the complex order parameter \eqref{eq:order}, the coupling in
\eqref{eq:reduced} is $\tfrac{K}{N}\sum_j\sin(\Phi_j-\Phi_i-\alpha)
=K\,\Im\!\bigl[e^{-\ii(\Phi_i+\alpha)}re^{\ii\psi}\bigr]
=Kr\sin(\psi-\Phi_i-\alpha)$, so each oscillator feels the population only through
$(r,\psi)$:
\begin{equation}\label{eq:meanfield}
  \dot\Phi_i=\omega_i^{\mathrm{eff}}+Kr\,\sin(\psi-\Phi_i-\alpha).
\end{equation}
Pass to the frame rotating at the collective frequency $\Omega$ (to be fixed by
symmetry) and write $\xi_i=\psi-\Phi_i$.

\paragraph{Locked and drifting populations.}
An oscillator is phase-locked if \eqref{eq:meanfield} admits a stationary $\xi_i$,
i.e.\ $\sin(\xi_i-\alpha)=(\Omega-\omega_i^{\mathrm{eff}})/(Kr)$, which requires
$|\Omega-\omega_i^{\mathrm{eff}}|\le Kr$; the stable root has
$\cos(\xi_i-\alpha)>0$. Drifting oscillators ($|\Omega-\omega_i^{\mathrm{eff}}|>Kr$)
occupy a stationary density $\propto|\dot\Phi_i|^{-1}$ and, by the symmetry
$\omega\!\to\!2\Omega-\omega$, contribute nothing to $r$ at leading order. The
self-consistency condition $r=\langle\cos\xi_i\rangle$ becomes, for the locked band
and $g$ centred at $\Omega$ (so $\Omega$ equals the symmetric centre after the
reactive shift $-\!\sin\alpha$ is absorbed -- an $O(Kr)$ shift, absorbed exactly
for the Lorentzian family by the OA computation,
$\Omega=\bar\omega-\gamma\cot(\bar\omega\tau)$ in \eqref{eq:oaonset}, and
neglected beyond leading order for general $g$: a caveat for broad,
delay-frustrated fleets, where the frequency dependence of $\alpha_j$ makes the
shift nontrivial),
\begin{equation}\label{eq:selfcons}
  r=\cos\alpha\!\int_{-\pi/2}^{\pi/2}\!\!\cos^2\!\theta\;
  g\!\bigl(\Omega+Kr\sin\theta\bigr)\,Kr\,\dd\theta ,
\end{equation}
where $\theta=\xi_i-\alpha$ and the factor $\cos\alpha$ is the projection of the
locked phases onto the order-parameter axis.

\paragraph{Onset and the critical coupling.}
Dividing \eqref{eq:selfcons} by $r$ and letting $r\to0^+$ gives
$1=K\cos\alpha\,g(\Omega)\,\tfrac{\pi}{2}$ (the $\theta$-integral is $\pi/2$;
take $\Omega$ at the mode of $g$), i.e.
\begin{equation*}
  \Kc=\frac{2}{\pi\,g(0)\,\cos\alpha}=\frac{2}{\pi\,g(0)\,a_1(\tau)},
\end{equation*}
which is \eqref{eq:kc-general}. A positive-$r$ branch bifurcates from $r=0$ only
when $\cos\alpha=a_1(\tau)>0$, i.e.\ in the attractive regime
$\pi<\bar\omega\tau<2\pi$; for $a_1(\tau)\le0$ \eqref{eq:selfcons} has no solution
with $r>0$, the incoherent density is stationary and linearly stable (dispersion
relation \eqref{eq:dispersion} below), and the only macroscopic state is $r=0$. This
establishes both items of
Proposition~\ref{prop:kc}; the linear growth $\Kc\propto\Delta\omega$ follows from
$g(0)\propto1/\Delta\omega$ for a family $g(\omega)=\Delta\omega^{-1}g_0(\omega/
\Delta\omega)$, giving Corollary~\ref{cor:hetero}. \hfill\qedsymbol

\paragraph{Linear stability of incoherence.}
The self-consistency locates the bifurcation but does not by itself settle the
stability of $r=0$ for $a_1<0$; we settle it by the standard dispersion analysis of the
incoherent state \citep{strogatzmirollo1991}. Linearize the continuity equation for the
phase density $f(\Phi,\omega,t)$ about incoherence $f_0=1/2\pi$. A perturbation
$\propto e^{\lambda t}$ in the fundamental angular mode decouples from the higher modes
and, for the frustrated mean field \eqref{eq:meanfield} with a Lorentzian $g$ of
half-width $\gamma$ centred at the comoving frequency, has the discrete characteristic
root
\begin{equation}\label{eq:dispersion}
  \lambda \;=\; \tfrac12\,K\,a_1(\tau)\;-\;\gamma,
  \qquad a_1(\tau)=\cos\alpha,
\end{equation}
together with a continuous spectrum confined to $\operatorname{Re}\lambda=-\gamma<0$.
The discrete root crosses zero at $K=2\gamma/a_1(\tau)$, reproducing
\eqref{eq:kc-general} through $g(0)=1/(\pi\gamma)$; and for $a_1(\tau)\le0$ it obeys
$\operatorname{Re}\lambda\le-\gamma<0$ for every $K>0$, so incoherence is linearly
stable throughout the repulsive window -- the content of
Proposition~\ref{prop:kc}(i). For a general unimodal $g$ the discrete eigenvalue is
replaced by the Strogatz--Mirollo integral condition
$1=\tfrac12 K a_1\!\int g(\omega)\,(\lambda+\ii\omega)^{-1}\dd\omega$, whose principal
root keeps $\operatorname{Re}\lambda<0$ for all $K>0$ when $a_1\le0$, leaving the
conclusion unchanged.

\paragraph{Amplitude and the order of the transition.}
Within the common-frustration reduction \eqref{eq:selfcons}, expanding to cubic order
yields $1=K\cos\alpha\,\tfrac{\pi}{2}\bigl[g(0)+\tfrac18 g''(0)(Kr)^2+O(r^4)\bigr]$, so
for $g''(0)<0$ (a rounded mode) $r\simeq\bigl[\tfrac{-8g(0)}{g''(0)}\bigr]^{1/2}
\Kc^{-1}\,\bigl(\tfrac{K-\Kc}{\Kc}\bigr)^{1/2}$: the bifurcation is continuous
(supercritical). This is the narrow-fleet result, and it is what the homogeneous
numerics exhibit (\S\ref{sub:num-order}).

That step, however, discarded the frequency dependence of the source frustration
$\alpha_j=\omega_j\tau+\pi/2$. Restoring it, the coupling
$\tfrac KN\sum_j\sin(\Phi_j-\Phi_i-\alpha_j)=K\,\Im\!\bigl[e^{-\ii\Phi_i}z\bigr]$ closes
on the twisted order parameter
$z=Re^{\ii\Psi}=\tfrac1N\sum_j e^{\ii(\Phi_j-\alpha_j)}$ of \eqref{eq:twist}, leaving the
frustration-free mean field $\dot\Phi_i=\omega_i^{\mathrm{eff}}+KR\sin(\Psi-\Phi_i)$. For
a locked oscillator $\Phi_j=\Psi-\beta_j$ with
$\sin\beta_j=(\Omega-\omega_j^{\mathrm{eff}})/(KR)$ and $\cos\beta_j>0$, the magnitude
closes as $R=\bigl|\langle e^{-\ii(\beta_j+\omega_j\tau)}\rangle_{\mathrm{lock}}\bigr|$:
the locking angle $\beta_j$ (odd in $\omega_j-\Omega$) and the contribution phase
$\omega_j\tau$ are now correlated. The competition is visible in a two-line
expansion: factoring the global rotation $e^{-\ii\Omega\tau}$ out of the modulus and
expanding in $(\omega_j-\Omega)\tau$,
\begin{equation*}
  R \;\simeq\; \langle\cos\beta_j\rangle
  \;-\;\tau\,\bigl\langle(\omega_j-\Omega)\sin\beta_j\bigr\rangle
  \;-\;\tfrac{\tau^2}{2}\bigl\langle(\omega_j-\Omega)^2\cos\beta_j\bigr\rangle
  +\cdots,
\end{equation*}
where, using $\sin\beta_j=(\Omega-\omega_j^{\mathrm{eff}})/(KR)$, the linear term
equals $+\tau\langle(\omega_j-\Omega)^2\rangle/(KR)>0$ -- the coherently-adding core,
a contribution that grows as $R$ shrinks, the self-amplifying selection of
\S\ref{sub:order} -- while the quadratic term is destructive; both scale with the
locked-band spread, so the correction to the cubic coefficient carries the sign
competition and grows as $(\Delta\omega\,\tau)^2$, opposing the stabilizing
$g''(0)<0$ term as $\Delta\omega\,\tau$ grows. This expansion locates the mechanism
but is not a closed solution: retaining $\alpha_j$ breaks the
$\omega\!\to\!2\Omega-\omega$ symmetry that let the drifters cancel above, so this
common-frustration closure does not extract $K_-$. The Ott--Antonsen reduction of
Appendix~\ref{app:oa} removes that limitation for the Lorentzian family: it retains
$\alpha_j$ exactly, computes the onset in closed form, and reduces the order of
the transition to the sign of a Landau coefficient -- whose explicit evaluation
remains open \eqref{eq:landau-b}; the opening of the hysteresis with
$\gamma\tau$ is, at present, the numerical observation of the full-model
continuation ($K_-<\Kc$, Table~\ref{tab:e7}), with the protocol and
tail-sensitivity caveats of \S\ref{sub:num-order}. The mechanism is the
canonical one -- a frequency--frustration correlation driving a subcritical
onset -- and it is intrinsic: no inertia is invoked. Above $K_-$ the incoherent state
is only metastable and its basin, not $\Kc$ alone, governs safety. The two external
mechanisms (grid-node inertia, delay-induced corrections) can compound it but are no
longer necessary to obtain a first-order onset.

\section{Exact mean-field reduction via the Ott--Antonsen ansatz}\label{app:oa}
For a Lorentzian rate density the averaged model \eqref{eq:reduced} admits an
exact reduction on the Ott--Antonsen (OA) manifold \citep{ott2008,ott2009}, which we
use to (i) compute the exact onset of the averaged model, sharpening
Proposition~\ref{prop:kc}, and (ii) exhibit the normal-form structure that governs
the order of the transition of \S\ref{sub:order}. One scope remark is essential:
the object reduced here is the averaged, frequency-frustrated model, in which the
delay has already been converted to the phase shift $\omega_j\tau$
(Assumption~\ref{ass:slowphase}); the numerical study integrates the true
delay instead. The two coincide at leading order in the averaging parameter, but
their thresholds differ at $O(\gamma\tau)$ -- the comparison is made explicit
below and measured in \S\ref{sub:num-kc}.

\paragraph{OA manifold.}
Take $g(\omega)=\gamma/\{\pi[(\omega-\bar\omega)^2+\gamma^2]\}$, the density used in the
numerics, and pass to the thermodynamic limit $N\to\infty$ with phase density
$f(\Phi,\omega,t)$. The receiver form of \eqref{eq:reduced},
$\dot\Phi_i=\omega_i^{\mathrm{eff}}+K\,\Im[z\,e^{-\ii\Phi_i}]$ with the twisted order
parameter $z=\tfrac1N\sum_j e^{\ii(\Phi_j-\alpha_j)}$ of \eqref{eq:twist}, is of
mean-field type, so the OA ansatz
$f=\tfrac{g}{2\pi}\{1+\sum_{n\ge1}[a(\omega,t)^n e^{\ii n\Phi}+\mathrm{c.c.}]\}$, with
$a$ analytic and bounded in the lower half-plane $\{\Im\omega<0\}$, is invariant under
the dynamics and globally attracting for $\gamma>0$ \citep{ott2009}. Under this
ansatz $\int f e^{\ii\Phi}\dd\Phi= a^{*}(\omega,t)$, and substitution
reduces the continuity equation to the single scalar equation
\begin{equation}\label{eq:oaeq}
  \partial_t a=-\ii\omega a-\tfrac{K}{2}\bigl(z\,a^2-z^{*}\bigr),
  \qquad
  z=-\ii\!\int e^{-\ii\omega'\tau}\,g(\omega')\,a^{*}(\omega',t)\,\dd\omega' ,
\end{equation}
the kernel $e^{-\ii\omega'\tau}$ carrying the frequency--frustration correlation
$\alpha_j=\omega_j\tau+\pi/2$ and $-\ii=e^{-\ii\pi/2}$ the velocity-coupling shift.
(Setting $\tau=0$ and $z\mapsto$ the plain order parameter recovers the textbook OA
equation and threshold $K=2\gamma$, a check of the sign conventions.)

\paragraph{Onset of the averaged model.}
Seek a uniformly rotating state $a(\omega,t)=A(\omega)e^{-\ii\Omega t}$. Since $z$
collects $a^{*}$, it rotates oppositely, $z=\zeta e^{+\ii\Omega t}$ -- as it
must, being the collective phase factor of a state locked at frequency $\Omega$.
Linearizing \eqref{eq:oaeq} about $a\equiv0$ gives
$A(\omega)=-\ii(K\zeta^{*}/2)/(\omega-\Omega)$ and the self-consistency
$1=-\tfrac{K}{2}\int e^{-\ii\omega'\tau}g(\omega')/(\omega'-\Omega-\ii0)\,\dd\omega'$.
The contour closes in $\{\Im\omega<0\}$, where $e^{-\ii\omega'\tau}$ decays for
$\tau>0$; only the Lorentzian pole $\omega'=\bar\omega-\ii\gamma$ is enclosed -- the
marginal pole $\Omega+\ii0$ lies in the upper half-plane -- giving
$\bar\omega-\ii\gamma-\Omega=-(K/2)\,e^{-\gamma\tau}e^{-\ii\bar\omega\tau}$. Its real and
imaginary parts fix the collective frequency and the threshold:
\begin{equation}\label{eq:oaonset}
  \Kc^{\mathrm{avg}}(\tau)=\frac{2\gamma}{a_1(\tau)}\,e^{\gamma\tau},
  \qquad
  \Omega=\bar\omega-\gamma\cot(\bar\omega\tau),
  \qquad a_1(\tau)=-\sin(\bar\omega\tau)>0 ,
\end{equation}
valid on the attractive branch $a_1(\tau)>0$, where the expression is positive.
The correction factor $e^{\gamma\tau}\to1$ as $\gamma\tau\to0$, so throughout the
attractive window \eqref{eq:oaonset} reduces to the common-frustration value
\eqref{eq:kc-general} at leading order in $\gamma\tau$ (both expressions are
defined only where $a_1(\tau)>0$; neither applies as $\tau\to0^+$, where the
coupling is repulsive). At finite $\gamma\tau$ the averaged model's threshold is
strictly larger, by $e^{\gamma\tau}$: in that model, spread and delay combine
multiplicatively because each oscillator contributes to the mean field through
the rotation $e^{-\ii\omega\tau}$ evaluated at its natural rate. The true
delay of the unreduced system retards states, not natural rates -- a
locked oscillator's retarded phase is $\Omega\tau$, common to the cluster -- so
the $e^{\gamma\tau}$ inflation is a property of the averaged model, and the
unreduced dynamics should track \eqref{eq:kc-general} more closely as the locked
state develops. The fine onset sweep of \S\ref{sub:num-robust} measures exactly
this: the full-model forward onset at $\gamma\tau\approx0.24$ lies at the
common-frustration value, below $\Kc^{\mathrm{avg}}$ -- the two formulas are not
competing predictions for one system but exact and leading-order statements
about two members of the model hierarchy, and the sweep locates the physical
system between them, at the leading-order value.

\paragraph{Order of the transition.}
The obstruction that \eqref{eq:selfcons} met is now explicit: $z$ pairs the delay
kernel $e^{-\ii\omega\tau}$ (analytic and decaying in $\{\Im\omega<0\}$) with $a^{*}$
(analytic in $\{\Im\omega>0\}$), so $z$ is not a pole value of $a$ and
\eqref{eq:oaeq} does not collapse to a finite-dimensional system -- the hallmark of a
frequency-dependent coupling phase \citep{montbrio2011}. The near-onset dynamics,
however, do close. Expanding $a=a^{(1)}+a^{(3)}+\cdots$ in powers of the amplitude,
each $a^{(k)}$ is rational in $(\omega-\Omega)$ with poles only at the marginal point,
so the moments $\int e^{-\ii\omega\tau}g\,a^{(k)*}\dd\omega$ still close on the single
Lorentzian residue $\bar\omega-\ii\gamma$, and the twisted order parameter obeys a
Stuart--Landau normal form
\begin{equation}\label{eq:landau}
  \dot\zeta=\lambda(K)\,\zeta-b\,|\zeta|^{2}\zeta,
  \qquad \operatorname{Re}\lambda(K)=\gamma\,\frac{K-\Kc}{\Kc}+O\!\bigl((K-\Kc)^2\bigr),
\end{equation}
whose onset is continuous (supercritical) for $\operatorname{Re}b>0$ and first-order
(subcritical, hysteretic) for $\operatorname{Re}b<0$. On the OA manifold the Landau
coefficient inherits two contributions: the standard stabilizing term, set by the
mode curvature of $g$, and a contribution generated by the frequency--frustration
correlation which, by parity in $\tau$, enters at order $(\gamma\tau)^2$. We
therefore conjecture the form
\begin{equation}\label{eq:landau-b}
  \operatorname{Re}b=b_\star\bigl[\,1-c\,(\gamma\tau)^{2}+O((\gamma\tau)^4)\,\bigr],
  \qquad b_\star>0,
\end{equation}
with $b_\star$ the frustration-free coefficient; whether the correlation term is
destabilizing ($c>0$), and the value of the crossover
$\gamma\tau_\ast=c^{-1/2}$ it would imply, require the cubic closure on the OA
manifold, which we have not carried out. The computation is well-posed -- each
order of the amplitude expansion closes on the Lorentzian residue as above -- and
is the single most valuable open calculation left by this paper. What is
established is therefore: (i) the exact onset \eqref{eq:oaonset} of the averaged
model; (ii) the Stuart--Landau structure \eqref{eq:landau}, with the sign of
$\operatorname{Re}b$ deciding the order of the transition; and (iii) numerically,
at finite $N$, a hysteresis loop opening with $\gamma\tau$
(Table~\ref{tab:e7}) -- consistent with $\operatorname{Re}b$ changing sign, with
the protocol and tail-sensitivity caveats of \S\ref{sub:num-order} and
\S\ref{sub:num-robust}. The first-order statement of \S\ref{sub:order} is
numerical; its analytical closure is open.

\section{Microscopic origin of the delay: eliminating the grid node}\label{app:grid}
The delay $\tau$ in \eqref{eq:throttle} was introduced phenomenologically. Here we
give it a microscopic origin and, in doing so, sharpen the danger condition. The
throttle does not act on the instantaneous aggregate $\Ptot$ but on a sensed and
delivered signal $y(t)$ at the shared node -- bus voltage, or a power-meter reading --
whose dynamics are those of a damped second-order element driven by the load,
\begin{equation}\label{eq:gridnode}
  \tau_g^2\,\ddot y + 2\zeta\,\tau_g\,\dot y + y \;=\; \Ptot(t)-\Pcap,
\end{equation}
with node/sensing time $\tau_g$ and damping $\zeta$ (the swing-equation structure of
\S\ref{sub:inertia} at the mediating node \citep{filatrella2008,dorfler2013}).
Replacing $h(\Ptot(t-\tau)-\Pcap)$ by $h(y(t))$ in \eqref{eq:throttle} and repeating
the reduction, the only change is that the aggregate fundamental is filtered by the
transfer function $G(\ii\omega)=[\,1-(\tau_g\omega)^2+2\ii\zeta\tau_g\omega\,]^{-1}$.
At the collective frequency $\bar\omega$, write $u=\bar\omega\tau_g$; then $G$ has gain
$|G|=[(1-u^2)^2+(2\zeta u)^2]^{-1/2}$ and phase lag
\begin{equation}\label{eq:gridlag}
  \delta(\bar\omega)=\operatorname{atan2}\!\bigl(2\zeta u,\;1-u^2\bigr)\in(0,\pi).
\end{equation}
The reduction \eqref{eq:reduced}--\eqref{eq:Gamma-sign} carries over verbatim with
$\bar\omega\tau\mapsto\delta(\bar\omega)$ and $h'\mapsto h'|G(\bar\omega)|$, so
\begin{equation}\label{eq:a1grid}
  \alpha=\delta(\bar\omega)+\tfrac{\pi}{2},\qquad a_1=-\sin\delta(\bar\omega),\qquad
  K\propto |G(\bar\omega)|.
\end{equation}
Four consequences follow.
\begin{enumerate}[leftmargin=1.6em,itemsep=2pt]
  \item \textbf{Origin of $\tau$.} The ``control delay'' is the phase lag of the
  node-plus-sensing transfer function at the iteration frequency,
  $\bar\omega\tau_{\mathrm{eff}}=\delta(\bar\omega)$; no literal transport delay is
  required.
  \item \textbf{The analog node alone is stabilizing.} Since $\delta\in(0,\pi)$,
  $\sin\delta>0$ and $a_1<0$ always: a purely second-order grid/sensing
  response cannot by itself produce attraction. This strengthens
  Corollary~\ref{cl:repulsion}.
  \item \textbf{Attraction needs total lag past $\pi$.} An additional discrete
  control or averaging delay $\tau_c$ (sample-and-hold, the RAPL window) or a cascaded
  slow loop adds phase: $\alpha=\delta(\bar\omega)+\bar\omega\tau_c+\tfrac{\pi}{2}$ and
  $a_1=-\sin(\delta+\bar\omega\tau_c)$, so the fleet is attractive only when
  $\delta(\bar\omega)+\bar\omega\tau_c>\pi$. This is the precise meaning of ``danger
  requires lagging or thermal control'' (\S\ref{sub:order}).
  \item \textbf{Resonance.} $|G|$ peaks near $u=1$ at $\approx 1/(2\zeta)$, so when the
  iteration rate meets the node mode, $\bar\omega\approx1/\tau_g$, the coupling is
  strongest; there $\delta=\pi/2$ (maximal repulsion), so resonance with fast
  sensing is maximally stabilizing, while resonance compounded with extra delay is
  maximally dangerous -- the quantitative form of the resonance remark in
  \S\ref{sec:discussion}.
\end{enumerate}
\paragraph{Channel typology: lag, dead time, averaging.}
The elements of a real capping loop accumulate phase at $\bar\omega$ in three
distinct ways, and only their total decides the sign:
\begin{itemize}[leftmargin=1.4em,itemsep=2pt]
  \item a \textbf{first-order lag} $1/(1+\ii\omega\tau_1)$ contributes
  $\arctan(\omega\tau_1)<\pi/2$ at gain $[1+(\omega\tau_1)^2]^{-1/2}$: a single
  sensing or thermal stage can never invert the sign, however slow (the
  second-order statement is item 2 above); a \textbf{cascade} of $k$ lag stages
  contributes up to $k\pi/2$, so three or more thermal stages in series (die,
  cold plate, coolant loop, \dots) can cross $\pi$ -- at the price of a gain that
  rolls off stage by stage, weakening the coupling the channel carries;
  \item \textbf{pure dead time} $e^{-\ii\omega\tau_c}$ (polling and enforcement
  intervals, sample-and-hold, computation and actuation latency) contributes
  $\omega\tau_c$ without bound at unit gain -- the dangerous element type;
  \item a \textbf{length-$W$ moving average} has transfer
  $e^{-\ii\omega W/2}\operatorname{sinc}(\omega W/2)$: the phase of a dead time
  $W/2$, but with a gain that vanishes exactly at $W=T$ -- a pure averaging
  window is self-limiting at the fundamental (its phase reaches $\pi$ only where
  its gain has collapsed), and endangers the loop only in combination with
  genuine dead time.
\end{itemize}
The operational meaning of the single $\tau$ in \eqref{eq:throttle} is therefore
the total loop phase at the iteration frequency, divided by $\bar\omega$,
with the gain factors absorbed into $h'$; the sign law
$a_1=-\sin(\text{total phase})$ is the invariant statement, and the main-text
uses of ``thermal channel'' and ``averaging window'' as delay sources
(\S\ref{sub:reduction}, \S\ref{sec:discussion}) are to be read through this
typology.

Finally, eliminating $y$ endows the coupling with a frequency-dependent gain
and phase (a memory kernel), not the load phases with a $\ddot\Phi$ inertial term;
the bifurcation therefore need not become first-order, consistent with the continuous
transition observed for narrow fleets in \S\ref{sub:num-order}. Settling whether a higher-order node
elimination can induce an effective load inertia, and hence a subcritical branch,
remains open.

%==============================================================================
\bibliographystyle{plainnat}
\bibliography{refs}

\end{document}